\documentclass[submission,copyright,creativecommons]{eptcs}
\providecommand{\event}{TARK 2019} 
\usepackage{breakurl}             
\usepackage{underscore}           

\usepackage{comment}

\usepackage{amsmath,amsthm,amssymb}
\usepackage{wrapfig}

\theoremstyle{definition}
\newtheorem{theorem}{Theorem}[section]
\newtheorem{fact}[theorem]{Fact}
\newtheorem{proposition}[theorem]{Proposition}

\newtheorem{example}[theorem]{Example}
\newtheorem{definition}[theorem]{Definition}
\newtheorem{lemma}[theorem]{Lemma}
  
\newtheorem{remark}[theorem]{Remark}

\usepackage{graphicx}

\title{Strategic Voting Under Uncertainty About the Voting Method\thanks{For helpful comments, we wish to thank Mikayla Kelley, the three anonymous referees for TARK, and an audience at the Center for Human-Compatible AI's Interdisciplinary Insights Seminar on May 1, 2019.}}
\author{Wesley H. Holliday
\institute{University of California, Berkeley}
\email{wesholliday@berkeley.edu}
\and
Eric Pacuit
\institute{University of Maryland}
\email{epacuit@umd.edu}
}
\def\titlerunning{Strategic Voting Under Uncertainty About the Voting Method}
\def\authorrunning{Holliday and Pacuit}
\begin{document}
\maketitle

\begin{abstract}Much of the theoretical work on strategic voting makes strong assumptions about what voters know about the voting situation. A strategizing voter is typically assumed to know how other voters will vote and to know the rules of the voting method. A growing body of literature explores strategic voting when there is uncertainty about how others will vote. In this paper, we study strategic voting when there is uncertainty about the voting method. We introduce three notions of manipulability for a set of voting methods: \textit{sure}, \textit{safe}, and \textit{expected} manipulability. With the help of a computer program, we identify voting scenarios in which uncertainty about the voting method may reduce or even eliminate a voter's incentive to misrepresent her preferences. Thus, it may be in the interest of an election designer who wishes to reduce strategic voting to leave voters uncertain about which of several reasonable voting methods will be used to determine the winners of an election.\end{abstract}

\section{Introduction}\label{IntroSection}

A well-known fact in the study of voting methods is that {\em strategic voting} cannot be avoided (see, e.g., \cite{Taylor2005,Meir2018}). A voter has an incentive to  vote strategically, or  {\em manipulate},  if she will achieve a preferable  outcome by misrepresenting her true preferences about the candidates.   There are two fundamental results showing that every reasonable voting method is susceptible to strategic voting. The Gibbard-Satterthwaite theorem \cite{Gibbard:73a,Satterthwaite:75a} shows that every voting method for three or more candidates that is resolute (i.e., always elects a single winner), unanimous (i.e., never elects a candidate $y$ if all voters rank another candidate $x$ above~$y$), and non-dictatorial (i.e., does not always elect the favorite candidate of some distinguished voter) can be manipulated. To study strategizing for irresolute voting methods, which may select more than one candidate as a winner,  additional assumptions are needed about the voters' rankings of {\em sets} of candidates.  The Duggan-Schwartz theorem  \cite{DugganSchwartz:2000} shows  that every voting method---whether resolute or irresolute---for three or more candidates that is non-imposed (i.e., any candidate can be elected) and has no nominator (i.e., no voter  can always ensure that her top ranked candidate is among the set of winners) can be manipulated by a {\em optimist} or {\em pessimist}, defined as a voter who compares sets of candidates in terms of her highest (resp.~lowest) ranked candidate in each set.  Related results by Kelly  \cite{kelly:1977}, Benoit \cite{benoit}, Feldman \cite{feldman}, and G\"ardenfors \cite{gardenfors}, among others, show that manipulation is unavoidable under different assumptions about   how to lift a voter's ranking of candidates to sets of candidates. 

In the theorems mentioned above, a voter's decision to manipulate relies on strong assumptions about what the voter {\em knows} about the voting situation.  Two key assumptions are:   (1) the strategizing voter knows how the other voters in the population will vote or have voted; (2) the strategizing voter knows and understands the rules of the voting method. There is a growing body of literature that explores weakening assumption (1) \cite{contizer-walsh-xia,osborne-rubinstein,chopra-pacuit-parikh,merrill,reijngoud-endriss,vanditmarsch-lang-sadine}.   One reason voters may be uncertain about how other voters will vote, even if they know other voters' true preferences, is because they  think that other voters  might  be voting strategically.   There are  sophisticated game-theoretic  models that explore  the implications of weakening assumption (1) for this reason (consult part II of \cite{Meir2018} for an overview of this literature).   Whatever the reason, when voters have partial information about how other voters will vote, the decision about whether  to manipulate is more complicated, since voters are uncertain about which sets of winners will result  from their decision to manipulate.  One natural assumption is that voters are risk averse and only  strategize when doing so cannot lead to a worse outcome (according to their true preferences) and might  lead to a better outcome \cite{contizer-walsh-xia}  (cf. \cite{reijngoud-endriss} for other assumptions about how voters decide to strategize).  

The potentially negative aspects of strategic voting (see, e.g., \cite{Satterthwaite1973,Conitzer2016}) have motivated the search for barriers against manipulation.\footnote{It has been argued that strategic voting also has positive aspects (see, e.g., \cite{Dowding2008}). In this paper, we will not debate whether the negative aspects outweigh the positive. Our point is about what barriers are available \textit{if} one wishes to reduce strategic voting.} One potential barrier, investigated mostly in the AI literature, is the computational complexity of determining a profitable manipulation for a given voting method (see, e.g., \cite{Faliszewski2010,Conitzer2016}). In this paper, we will investigate another potential barrier against strategic voting: adding uncertainty about the voting method that will be used to determine the set of winners. Thus, we weaken assumption (2) above. Suppose that $S$ is a set of voting methods.  We assume that the voters know that one of the methods from $S$ will be used to select the winners, but they do not know which one---hence we call $S$ the \textit{uncertainty set}.  To simplify the initial study, in this paper we  assume that a strategizing voter knows how the other voters will vote or have voted.  This assumption is not unrealistic in certain voting situations; for example, in a hiring committee meeting in which committee members are asked to sequentially report their rankings of job candidates, the committee member last in the sequence will already know how all the other members have ranked the candidates. 

As in the situation where a voter has partial information about how other voters will vote, in the situation where a voter has uncertainty about the voting method, the voter must take into account different possible winning sets of candidates when deciding whether to strategize.  Given a ranking of sets of candidates, we consider three ways for a voter to solve this decision problem.  From the most conservative to the least conservative, the three types of manipulation  are: 
 
 \begin{enumerate}
 \item {\em Sure manipulation}: It is certain that submitting an insincere ranking will lead to a better outcome no matter which voting method is used. 
  
 \item {\em Safe manipulation}: It is certain that submitting an insincere ranking will lead to an outcome that is at least as good and might lead to a better outcome.\footnote{In \cite{Ianovski2011} and \cite{Slinko2014} the term `safe manipulation' is used in a different sense and in a different context involving manipulation by coalitions of voters.}   
 
 \item {\em Expected manipulation}: Given a probability distribution on the set of voting methods, submitting an insincere ranking is more likely to lead to a better outcome than to lead to a worse outcome.å
 \end{enumerate}
 We aim to find combinations of voting methods that reduce or even eliminate a voter's incentive to strategize (given one of the above notions of manipulation and a ranking of sets of candidates).     That is, does uncertainty about the voting method reduce the chance that voters will strategize? 

If so, it may be in the interest of an election designer who wishes to reduce strategic voting to leave voters uncertain about which of several voting methods will be used to determine the winners of an election. Though perhaps implausible in the case of a democratic political election, creating such uncertainty does not seem an implausible choice by, e.g., the Chair of a hiring committee. Moreover, if committee members judge each of the possible voting methods to be reasonable and also wish to reduce strategic voting, they may well endorse the choice of the Chair to create such uncertainty.\footnote{Naturally, the Chair should be required to fix the voting method to be used \textit{before} seeing the votes, so that the Chair cannot pick the voting method that produces his or her favored outcome based on the already submitted votes.}

Another way that voters may be uncertain about the outcome of an election is that a voting method may use randomization to select the winners.  A {\em probabilistic voting method} assigns a lottery over the set of candidates to each collection of rankings of the candidates for the voters (see \cite{brandt-survey} for a recent survey of results about probabilistic voting methods).    Gibbard \cite{gibbard1977} noted that a {\em random dictatorship} is immune to strategizing.\footnote{Cf.~\cite{Hylland1980}, where voters are assumed to submit utility functions over the set of candidates. In this paper, we restrict attention to the case where voters submit rankings of the candidates.}   Suppose that $S$ is a set of dictatorships---one for each voter. After the voters submit their rankings, one of the dictatorships in $S$ is chosen, and the candidate  ranked first by the dictator is selected as the winner.   Assuming voters submit rankings of the candidates without ties, this method always selects a single winner.   It is clear that no voter has an incentive to misrepresent their top choice, since either they will not be chosen as the dictator, in which case their ranking will be ignored, or they will be chosen as the dictator, in which case their top choice is guaranteed to win.  Of course, each dictatorship is immune to strategizing by itself.   Beyond this example of random dictatorship, there are general results suggesting that randomization can be an effective barrier to manipulation \cite{nunez-pivato,aziz-brandl-brandt-brill}.   

The relationship between our work and probabilistic voting methods is clarified at the end of this paper. In short, although any uncertainty set $S$ of voting methods determines a probabilistic voting method $F_S$, none of our three notions of manipulation above with respect to $S$ is equivalent to a standard notion of manipulation with respect to the probabilistic voting method $F_S$. In addition, since the mapping $S\mapsto F_S$ from uncertainty sets of voting methods to probabilistic voting methods is not one-to-one, there is no obvious definition of sure, safe, or even expected manipulation for probabilistic voting methods such that for any uncertainty set $S$ of voting methods, $S$ is susceptible to sure/safe/expected manipulation if and only if $F_S$ is susceptible to sure/safe/expected manipulation. Thus, there is no obvious way to reframe our investigation purely in the language of probabilistic voting methods.

Both our approach and that of probabilistic voting methods involve a loss of transparency at the time of the vote. However, the approaches seem to differ in the \textit{explainability} of the outcome after the vote. In our approach, the election designer may simply inform voters after the vote of which voting method $f$ from $S$ was used to determine the outcome of the election, so the algorithm for $f$ can be used to explain why the election had one outcome rather than another. This may be a more intelligible reason for voters than ``this was the outcome of the lottery.'' Abstractly, the difference is that in our setting, initially the voters have subjective uncertainty about which deterministic mechanism will be used; but after this subjective uncertainty is removed, there is an explanation of why the election had one outcome rather than another in terms of the deterministic mechanism applied to the voters' inputs.\footnote{That the voters have subjective uncertainty about which method in $S$ will be used does not imply that the election designer uses a chance process to determine which method will be used. As in standard voting theory, we make no assumption about how the election designer chooses the voting method, whether by randomization, consideration of axiomatic properties, etc.} By contrast, with a probabilistic voting method, the only available ``explanation'' of the outcome is in terms of an inherently stochastic mechanism applied to the voters' inputs, which may fail to explain why the election had one outcome rather than another with non-zero probability (cf.~\cite[p.~24]{Jeffrey1971}, \cite[p.~238]{Railton1981}).\footnote{In the case where $f$ is an irresolute voting method that outputs a winning set $X$ of candidates in a given election, a single winner may be chosen by lottery; but still we have an explanation in terms of the algorithm for $f$ of why the field was narrowed to $X$ rather than some other winning set before applying tiebreaking.} Of course, whether such a contrastive explanation is desirable may vary from case to case.

Our findings in this paper are of two main types: analytic results with mathematical proofs and data from computer searches. For these searches we generalize to sets of voting methods a standard index of manipulability for a single voting method, known as the Nitzan-Kelly index \cite{Nitzan1985,Kelly1985}, which gives the percentage of profiles with $n$ candidates and $m$ voters in which at least one voter in the profile has an incentive to manipulate. In probabilistic terms, assuming the Impartial Culture Model \cite{Guilbaud1952} in which every profile is equally probable, this index gives the probability that in a randomly chosen profile, at least one voter has an incentive to manipulate. Similarly, we report data on the percentage of profiles in which at least one voter has an incentive to manipulate against a set $S$ of voting methods. All of the data in the paper were produced by a Python script available at \href{https://github.com/epacuit/strategic-voting}{https://github.com/epacuit/strategic-voting}.

The rest of the paper is organized as follows.   The next section introduces our formal framework, including the definitions of the voting methods we study in this paper.  The three notions of manipulation are studied in Sections \ref{Section:SureManipulation}, \ref{Section:SafeManipulation}, and \ref{Section:ExpectedManipulation}, respectively.    Section \ref{Section:ProbabilisticSocialChoice} discusses the connections with probabilistic social choice.  Finally, in Section \ref{Section:Conclusion} we conclude with some pointers to future work. 

 \section{Preliminaries}\label{Section:Preliminaries}
 
Let $C$ be a  nonempty finite set of \textit{candidates} and $V$ a nonempty finite set of \textit{voters}.    We use lower case letters from the beginning and end of the alphabet $a, b, c, \ldots, x, y, z, \ldots$   for    elements of $C$ and lower case letters from the middle of the alphabet  $i, j, k,\ldots$ for  elements of $V$. 

   A voter's {\em  ranking} of the set of candidates   is a strict linear order $P$ on $C$.  Let $L(C)$ be the set of all strict linear orders on $C$.  For $P\in L(C)$ and $X\subseteq C$,  let $\mathrm{max}(X, P)$ be the maximally ranked element of $X$, i.e.,  $\mathrm{max}(X, P) = y$ where for all $x\in X$, if $x\ne y$, then  $y\mathrel{P} x$ (such an element always exists since $P$ is assumed to be a strict linear order). Similarly, let $\mathrm{min}(X, P)$ be the minimally ranked element of $X$, i.e.,  $\mathrm{min}(X, P) = y$ where for all $x\in X$, if $x\ne  y$, then $x\mathrel{P} y$.  We say $y\in C$ is ranked $r$th by $P$ when $|\{x\in C\ |\ x\mathrel{P} y\}| = r-1$.\footnote{As usual, for a set $A$, $|{A}|$ is the number of elements in ${A}$.}    So, for example,  $y$ is ranked 1st by $P$ if and only if  $\mathrm{max}(C, P)=y$.  To simplify our notation, we specify a ranking by simply listing candidates from highest to lowest in the ranking, e.g., $a \, b \, c \, d$ for the ranking $a\mathrel{P} b\mathrel{P} c\mathrel{P} d$. 
   
A \textit{profile $\mathbf{P}$ for $(C,V)$} is an element of $L(C)^V$, i.e., a function assigning to each $i\in V$ a relation $\mathbf{P}_i\in L(C)$. If $|C|=n$ and $|V|=m$, we call a profile for $(C,V)$ an \textit{$(n,m)$-profile}. A \textit{pointed profile} for $(C,V)$ is a pair $(\mathbf{P},i)$ where $\mathbf{P}$ is a profile and $i\in V$. For $x,y\in C$, let $\mathbf{P}(x,y)=\{i\in V\mid x\mathbf{P}_i y\}$. We write $\mathbf{N}_\mathbf{P}(x,y)$ for the number of  voters in $\mathbf{P}$ ranking $x$ above $y$, i.e., $\mathbf{N}_\mathbf{P}(x,y) = |\mathbf{P}(x,y)|$.    We say that a {\em majority prefers $x$ to $y$} in $\mathbf{P}$, denoted $x >^M_\mathbf{P} y$,  when $\mathbf{N}_\mathbf{P}(x,y) > \mathbf{N}_\mathbf{P}(y,x)$.     Let $Net_\mathbf{P}(x,y) = \mathbf{N}_{\mathbf{P}}(x,y) - \mathbf{N}_{\mathbf{P}}(y,x)$.   So $x >_\mathbf{P}^M y$ if and only if $Net_\mathbf{P}(x,y) > 0$.  Finally, from the strict linear order $\mathbf{P}_i$ we define the weak relation $\mathbf{R}_i$ by $x\mathbf{R}_iy$ iff $x\mathbf{P}_iy$ or $x=y$.

A  \textit{voting method} for $(C, V)$ is a function assigning a nonempty subset of candidates, called the \textit{winning set},  to each profile, i.e., $f:L(C)^V\rightarrow\wp(C)\setminus \{\varnothing\}$.  The following are the voting methods we will discuss in this paper (also see, e.g., \cite{Pacuit2019}). 

(1) Positional scoring rules:  Suppose $\langle s_1, s_2, \ldots, s_m\rangle$ is a vector of numbers, called a {\em scoring vector}, where for each $l=1,\ldots, m-1$, $s_l \ge s_{l+1}$. Suppose ${P}\in L(C)$.   The {\em score of $x\in C$ given $P$} is  $score(P,x)=s_r$ where $r$ is the rank of $x$ in $P$. For each profile $\mathbf{P}$ and $x\in C$, let $score(\mathbf{P},x)= \sum_{i=1}^n score(\mathbf{P}_i, x)$.   A positional scoring rule for a scoring vector $\vec{s}$ assigns to each profile $\mathbf{P}$  the set of candidates that maximize their score according to $\vec{s}$ in $\mathbf{P}$.   That is, a voting method $f$ is a positional scoring rule for a scoring vector $\vec{s}$ provided that for all $\mathbf{P}\in L(C)^V$, $f(\mathbf{P})=\mathrm{argmax}_{x\in C} score(\mathbf{P}, x)$.   We study two such rules:

\smallskip 

$\mathtt{Borda}$:  the positional scoring rule for ${\langle n-1, n-2, \ldots, 1,  0\rangle}$.

\smallskip 

$\mathtt{Plurality}$: the positional scoring rule for $\langle 1, 0, \ldots,   0\rangle$.
\smallskip 

(2) The {\em Condorcet winner} in a profile $\mathbf{P}$ is a candidate $x\in C$ that is the maximum of the majority ordering, i.e., for all $y\in C$, if $x\ne y$,  then $x>_\mathbf{P}^M y$.   The Condorcet voting method is:

$\mathtt{Condorcet}(\mathbf{P})=\begin{cases} \{x\} & \text{if $x$ is the Condorcet winner in $\mathbf{P}$}\\
C & \text{if there is no Condorcet winner}.
\end{cases}$

(3) The {\em win-loss} record for a candidate $x\in C$ in a profile $\mathbf{P}$ is the number of candidates $z$ such that a majority prefers  $x$ to $z$ in $\mathbf{P}$ minus the   number of candidates $z$ such that a majority prefers  $z$ to $x$ in $\mathbf{P}$.   Formally, for each $\mathbf{P}$ and $x\in C$, let $wl_\mathbf{P}(x) = |\{z \ |\ Net_\mathbf{P}(x,z) > 0\}| - |\{z\ |\ Net_\mathbf{P}(z,x) > 0\}|$. The $\texttt{Copeland}$ winners are the candidates with maximal win-loss records:  $\mathtt{Copeland}(\mathbf{P}) =  \mathrm{argmax}_{x\in C}(wl_\mathbf{P}(x))$. 
 
(4) The \textit{support} for a candidate $x\in C$ in a profile $\mathbf{P}$ is found by calculating for each candidate $y\neq x$ the number of voters who rank $x$ above $y$ and then taking the minimum of these values. Formally, for each $\mathbf{P}$ and $x\in C$, let $supp(x,\mathbf{P}) = \mathrm{min}(\{\mathbf{N}_\mathbf{P}(x,y)\ |\ y\in C, y\ne x\})$. The \texttt{MaxMin} (also known as Simpson's Rule) winners are the candidates with maximal support: $\texttt{MaxMin}(\mathbf{P}) = \mathrm{argmax}_{x\in C}(supp(x,\mathbf{P}))$.

(5) $\texttt{PluralityWRunoff}$: Calculate the plurality score for each candidate---the number of voters who rank the candidate first. If there are 2 or more candidates with the highest plurality score, remove all other candidates and select the $\texttt{Plurality}$ winners from the remaining candidates. If there is one candidate with the highest plurality score, remove all candidates except the candidates with the highest or second-highest plurality score, and select the $\texttt{Plurality}$ winners from the remaining candidates.
 
(6)  $\texttt{Hare}$:  Iteratively remove all candidates with the fewest number of voters who rank them first, until there is a candidate who is a \textit{majority  winner}, i.e., ranked first by a majority of voters.   If, at some stage of the removal process, all remaining candidates have the same number of  voters who rank them first (so all candidates would be removed), then all remaining candidates  are selected as winners.

(7) $\texttt{Coombs}$:  Iteratively remove all candidates with the most  number of voters who rank them last, until there is a candidate who is a majority  winner.   If, at some stage of the removal process, all remaining candidates have the same number voters who rank them last (so all candidates would be removed), then all remaining candidates  are selected as winners.

(8)  $\texttt{Baldwin}$:  Iteratively remove all candidates with the smallest \texttt{Borda} score, until there is a single candidate remaining.   If, at some stage of the removal process, all remaining candidates have the same \texttt{Borda} score (so all candidates would be removed), then all remaining candidates  are selected as winners.

(9)  Rather than removing candidates with the lowest \texttt{Borda} score, the next two methods remove all candidates who have a \texttt{Borda} score below the average \texttt{Borda} score for all candidates.    There are two versions of this voting method \cite{niou}: $\texttt{StrictNanson}$ iteratively removes all candidates whose \texttt{Borda} score is strictly smaller than the average \texttt{Borda} score (of the candidates remaining at that stage), until one candidate remains.   $\texttt{WeakNanson}$ iteratively removes all candidates whose \texttt{Borda} score is less than or equal to the average \texttt{Borda} score (of the candidates remaining at that stage), until one candidate remains.  If, at some stage of the removal process, all remaining candidates have the same \texttt{Borda} score (so all candidates would be removed), then all remaining candidates  are selected as winners.

\begin{definition} Let $\mathsf{Methods}$ be the set of 11 voting methods described above.
\end{definition}

We are interested in situations in which the strategizing voter is uncertain about which voting method will be used to determine the winner(s). As noted in Section \ref{IntroSection}, we represent such uncertainty by a set $S$ of voting methods, called the \textit{uncertainty set}. This is the set of voting methods $f$ such that it is consistent with the strategizing voter's knowledge that $f$ will be~used.

Each of the above methods may select more than one winner for a given profile. Thus, to discuss strategizing, one needs a notion of when one set of candidates is ``preferable'' to another for a particular voter. The following are the standard notions from the literature on strategic voting (see, e.g., \cite{Taylor2005}). 
 
\begin{definition}\label{DominanceNotions} Let $\mathbf{P}$ be a profile, $i\in V$, and $X,Y\subseteq C$. We define the following \textit{dominance notions}, each of which has a nonstrict ($\geq$) and strict ($>$) version: \\[2pt]
\noindent \textit{weak}: (a) $X\geq_{\mathbf{P}_i}^{weak} Y$ iff $\forall x\in X$ $\forall y\in Y$: $x \mathbf{R}_i y$; \quad (b) $X>_{\mathbf{P}_i}^{weak} Y$ iff  $X\geq_{\mathbf{P}_i}^{weak} Y$ and $\exists x\in X$ $\exists y\in Y$: $x \mathbf{P}_i y$.

\noindent \textit{optimistic}: (a) $X\geq_{\mathbf{P}_i}^{Opt} Y$ iff $\mathrm{max} (X,\mathbf{P}_i) \mathrel{\mathbf{R}_i} \mathrm{max} (Y,\mathbf{P}_i)$;\quad (b) $X>_{\mathbf{P}_i}^{Opt} Y$ iff $\mathrm{max} (X,\mathbf{P}_i) \mathrel{\mathbf{P}_i}  \mathrm{max} (Y,\mathbf{P}_i)$.

\noindent \textit{pessimistic}: (a) $X\geq_{\mathbf{P}_i}^{Pes} Y$ iff $\mathrm{min} (X,\mathbf{P}_i) \mathrel{\mathbf{R}_i} \mathrm{min} (Y,\mathbf{P}_i)$;\quad (b) $X>_{\mathbf{P}_i}^{Pes} Y$ iff $\mathrm{min} (X,\mathbf{P}_i) \mathrel{\mathbf{P}_i} \mathrm{min} (Y,\mathbf{P}_i)$.

\end{definition}

\section{Sure manipulation}\label{Section:SureManipulation}

If a voter is uncertain about which voting method from a set $S$ of voting methods will determine the winners of an election, the most conservative approach to strategic voting is to submit an insincere ranking if and only if the voter is sure that by doing so, the set of winners will be strictly better from the point of view of her true ranking (and the relevant dominance notion) \textit{no matter which voting method from $S$ is used}. This approach is appropriate when there is a cost to submitting an insincere ranking that a voter is only willing to incur if it will surely improve the set of winners. For example, in the context of a hiring committee in which committee members know each other's true preferences over the candidates, e.g., through deliberation, there may be a social cost in submitting an insincere ranking of the candidates, which a committee member is willing to bear only if doing so is sure to result in a preferable set of winners. These considerations motivate the following definition.

\begin{definition} Let $(\mathbf{P},i)$ be a pointed profile, $\Delta$ a dominance notion, and $S$ a set of voting methods. We say that $(\mathbf{P},i)$ \textit{witnesses sure $\Delta$-manipulability for $S$} if and only if there is a profile $\mathbf{P}'$ differing from $\mathbf{P}$ only in $i$'s ranking such that $\forall f\in S: f(\mathbf{P}')>_{\mathbf{P}_i}^\Delta f(\mathbf{P})$. We then say that $(\mathbf{P},i)$ witnesses sure $\Delta$-manipulability for $S$ \textit{by transitioning to $\mathbf{P}'$} and that $i$ \textit{has a sure $\Delta$-manipulation incentive under $S$ to transition to $\mathbf{P}'$}. A profile $\mathbf{P}$ \textit{witnesses sure $\Delta$-manipulability for $S$} if and only if there is an $i\in V$ such that $(\mathbf{P},i)$ witnesses sure $\Delta$-manipulability for $S$. Finally, we say that $S$ is \textit{susceptible to sure $\Delta$-manipulation for $(n,m)$} if and only if there is anå $(n,m)$-profile $\mathbf{P}$ that witnesses sure $\Delta$-manipulability for~$S$.
\end{definition}

As an initial example, consider the $\texttt{Borda}$ method. Recall that $\texttt{Borda}$ is not resolute, so the $\texttt{Borda}$ winning set for a particular profile may contain multiple candidates. To obtain a resolute method from an irresolute method $f$ such as $\texttt{Borda}$, we may fix a tiebreaking mechanism, understood as a strict linear order $L$ on $C$, and define  $f_L$ to be the resolute voting method defined by $f_L(\mathbf{P})= \mathrm{max}(f(\mathbf{P}), L)$. An especially natural example of uncertainty about the voting method arises when there is uncertainty about the tiebreaking mechanism to be used for a fixed voting method. Moreover, such uncertainty can be a barrier to sure manipulation. For example, for $n=3$, $\{\texttt{Borda}\}$ is susceptible to sure weak dominance manipulation, but this is not so when there is complete uncertainty about the tiebreaking mechanism.

\begin{proposition} For any $m\geq 4$, $\{\texttt{Borda}_L \mid L\mbox{ a linear order on }C\}$ is not susceptible to sure weak dominance manipulation for $(3,m)$.
\end{proposition}
\begin{proof} Suppose for contradiction that there is a pointed profile $(\mathbf{P},i)$ that witnesses sure weak dominance manipulation for $S=\{\texttt{Borda}_L \mid L\mbox{ a linear order on }C\}$ by transitioning to $\mathbf{P}'$. By the choice of $S$, it follows that (i) every candidate in the \texttt{Borda} winning set for $\mathbf{P}'$ is strictly preferred according to $\mathbf{P}_i$ to every candidate in the \texttt{Borda}  winning set in $\mathbf{P}$. Let $i$'s ranking of the candidates in $\mathbf{P}$ be $\alpha\beta\gamma$.

First, the $\mathtt{Borda}$ winning set in $\mathbf{P}$ cannot contain $\alpha$. For if it does, then there is no incentive to manipulate under $\mathtt{Borda}_{\alpha > \beta > \gamma}$ or $\mathtt{Borda}_{\alpha > \gamma > \beta}$, contradicting sure weak dominance manipulation. Second, the $\mathtt{Borda}$ winning set in $\mathbf{P}$ cannot contain $\gamma$. For if it does, then by (i), the $\mathtt{Borda}$ winning set in $\mathbf{P}'$ cannot contain $\gamma$; but there is no $\mathbf{P}'$ differing from $\mathbf{P}$ only in $i$'s ranking such that $\gamma$ is in the $\mathtt{Borda}$ winning set in $\mathbf{P}$ but not in the $\mathtt{Borda}$ winning set in $\mathbf{P}'$. Finally, we claim that the $\mathtt{Borda}$ winning set in $\mathbf{P}$ cannot be $\{\beta\}$. Suppose it is. By (i), the $\mathtt{Borda}$ winning set in $\mathbf{P}'$ cannot contain $\beta$ or $\gamma$, so it must be $\{\alpha\}$. But this contradicts the fact that for 3 candidates, $\mathtt{Borda}$ is not single-winner manipulable \cite[p.~57]{Taylor2005}. Having ruled out all possible $\mathtt{Borda}$ winning sets in $\mathbf{P}$, we obtain a contradiction.
\end{proof}

In this paper, we will focus on the case where $S$ contains different voting methods, as in the following example, rather than the same voting method with different tiebreaking rules.  

\begin{example} Consider the following $(3,4)$-profile $\mathbf{P}$ for $C=\{a, b, c\}$ and $V=\{1,2,3,4\}$:\\

\begin{minipage}{1in}
\begin{center}
\begin{tabular}{c|c|c|c}
$1$ & $2$ & $3$ & $4$ \\\hline
$a$ & $b$ & $c$ & $c$\\
$b$ & $c$ & $a$ & $b$\\
$c$ & $a$ & $b$ & $a$\\
\end{tabular}
\end{center}
\end{minipage}\hspace{.2in}\begin{minipage}{4.5in}
The winning sets for the voting methods are: 
\begin{itemize}
\item $\{c\}$  for $\texttt{Borda}$, $\texttt{Copeland}$, $\texttt{Hare}$, \texttt{WeakNanson}, $\texttt{Plurality}$, $\texttt{PluralityWRunoff}$; 
\item  $\{b, c\}$ for \texttt{Baldwin}, $\texttt{Coombs}$, \texttt{MaxMin}, and $\texttt{StrictNanson}$.
 \end{itemize}
 \end{minipage}\\
 
 \noindent If candidate $1$ changes her ranking to $b \, a \, c$, then all the methods select the winning set $\{b, c\}$.     Suppose that $\Delta$ is either weak or optimistic dominance. Since  $\{b, c\}>_{\mathbf{P}_1}^{\Delta} \{c\}$, $(\mathbf{P}, 1)$ witnesses  sure $\Delta$-dominance manipulability for any nonempty subset of methods from $\{$\texttt{Borda}, \texttt{Copeland}, \texttt{Hare}, \texttt{WeakNanson}, \texttt{Plurality}, \texttt{PluralityWRunoff}$\}$.
Since $\{b, c\}\not >_{\mathbf{P}_1}^{\Delta} \{b, c\}$,  this pointed profile $(\mathbf{P}, 1)$ does not witness sure $\Delta$-dominance manipulability  for any set of voting methods that contains one or more of the following methods: 
\texttt{Baldwin}, \texttt{Coombs}, \texttt{MaxMin}, \texttt{StrictNanson}. 
Of course, there may be other pointed profiles witnessing sure $\Delta$-manipulation for sets of methods containing one of these methods.  

 The above profile $\mathbf{P}$ does not witness sure pessimist manipulation for any set of voting methods.  But the following profile $\mathbf{P}'$ witnesses sure pessimist manipulation:\\
 
 \begin{minipage}{1in}
\begin{center}
\begin{tabular}{c|c|c|c}
$1$ & $2$ & $3$ & $4$ \\\hline
$a$ & $a$ & $b$ & $c$\\
$b$ & $c$ & $a$ & $b$\\
$c$ & $b$ & $c$ & $a$\\
\end{tabular}
\end{center}
\end{minipage}\hspace{.2in}\begin{minipage}{4.5in}The winning sets for the voting methods are: 
\begin{itemize}
\item $\{a\}$  for $\texttt{Borda}$, $\texttt{Copeland}$, $\texttt{Hare}$, \texttt{WeakNanson}, $\texttt{Plurality}$, and $\texttt{PluralityWRunoff}$;
\item  $\{a, b\}$ for \texttt{Baldwin}, $\texttt{Coombs}$, \texttt{MaxMin}, and $\texttt{StrictNanson}$.
\end{itemize}
\end{minipage}\\

\noindent If candidate $1$ changes her ranking to $a  \, c \, b$, then all the methods select the winning set $\{a\}$.      
Since  $\{a\}>_{\mathbf{P}_1'}^{Pes} \{a,b\}$ and $\{a,b\}\not >_{\mathbf{P}_1'}^{Pes} \{a,b\}$, $(\mathbf{P}, 1)$ witnesses  sure pessimist manipulability for any nonempty subset of methods from
$\{\texttt{Baldwin}, \texttt{Coombs}, \texttt{MaxMin}, \texttt{StrictNanson}\}$, 
but no set   containing any of the following methods:
\texttt{Borda}, \texttt{Copeland}, \texttt{Hare}, \texttt{WeakNanson}, \texttt{Plurality}, \texttt{PluralityWRunoff}.

\end{example}

\subsection{Eliminating sure manipulation with a pair of voting methods}

Our first question is whether  uncertainty about the voting method may {\em eliminate} sure manipulation. 

\begin{definition}\label{eliminate-sure-dominance} Let $S$ be a set of voting methods.  We say that $S$   \textit{eliminates sure $\Delta$-manipulation for $(n,m)$} iff $S$ is not susceptible to sure $\Delta$-manipulation for $(n,m)$ but every nonempty $S'\subsetneq S$  is susceptible to sure $\Delta$-manipulation for $(n,m)$. 
\end{definition}

\begin{figure}[h]
\includegraphics[scale=0.28]{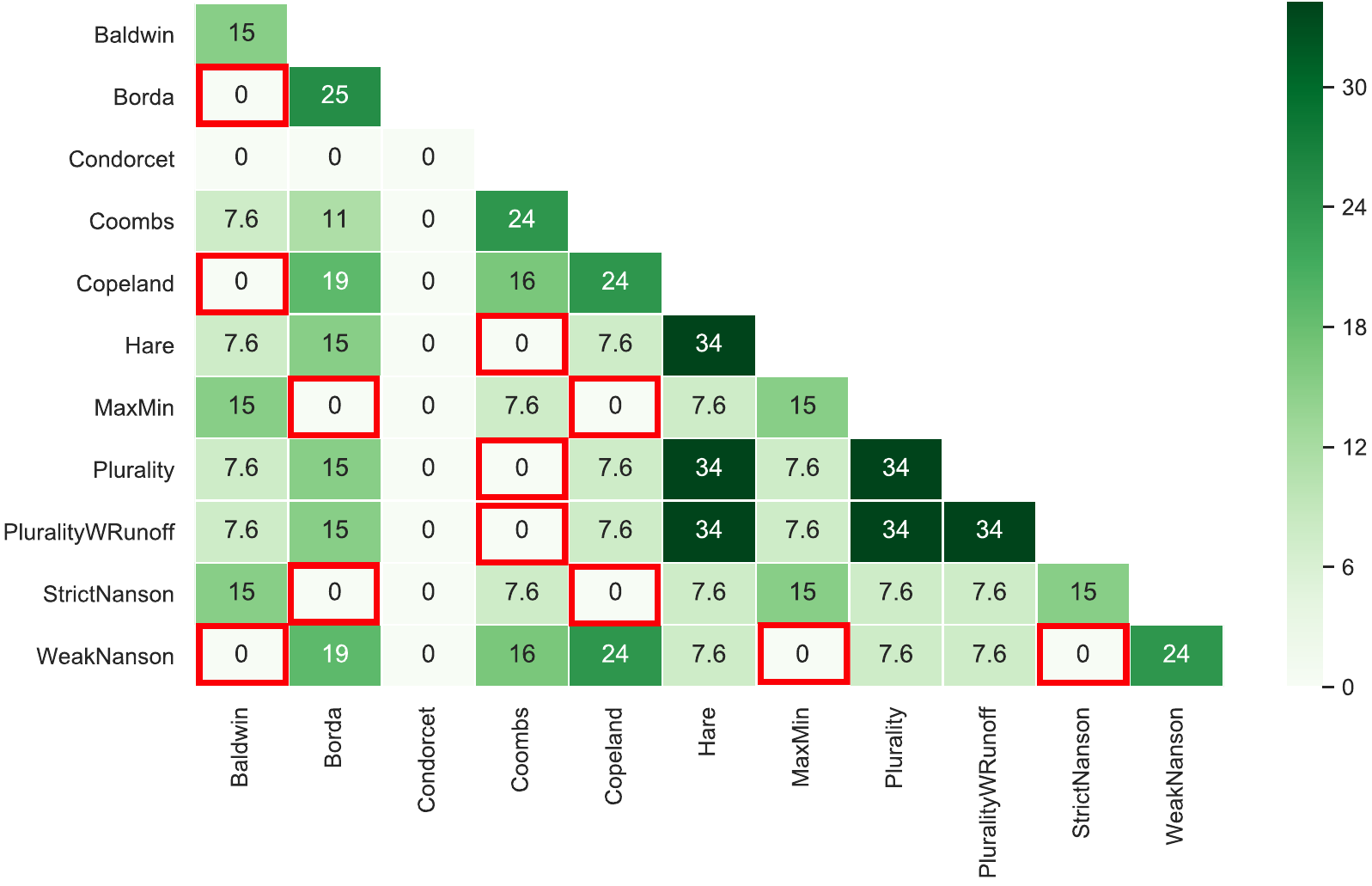}\qquad
\includegraphics[scale=0.28]{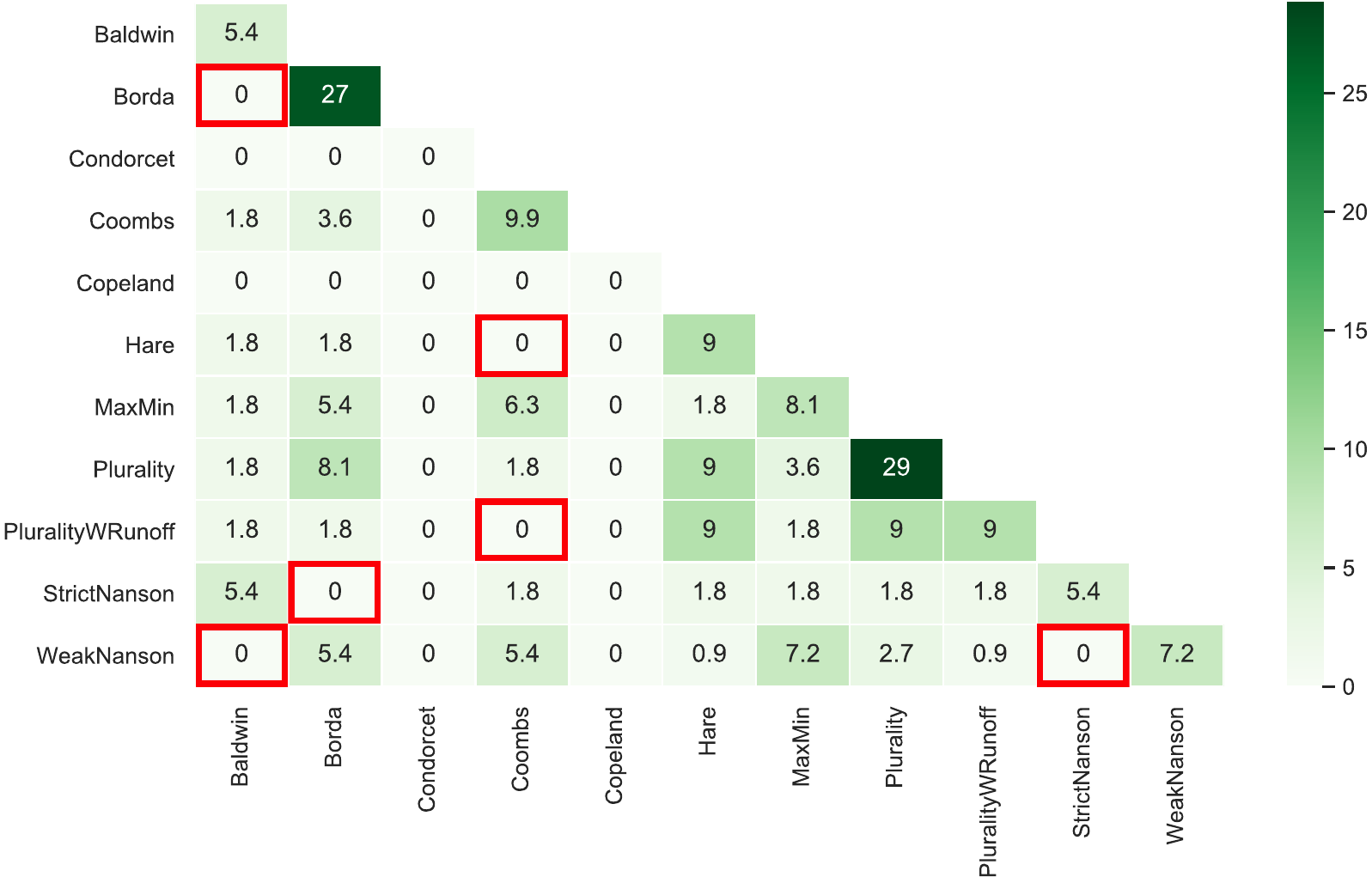}
\caption{{\small Percentage of profiles witnessing sure weak dominance manipulation for $(3,4)$ (left) and $(3,7)$ (right).}}\label{3437sureweak}
\end{figure}

The numbers in Figure~\ref{3437sureweak} represent the percentage of profiles that witness sure weak dominance manipulation for $(3,4)$ and $(3,7)$. For example, 7.6\% of  $(3,4)$-profiles witness sure weak dominance manipulation for $\{\texttt{Plurality}, \texttt{Copeland}\}$.   The numbers along the diagonal give the percentage of profiles witnessing weak dominance manipulation for a single voting method.\footnote{For 3 candidates, \texttt{Hare} and \texttt{PluralityWRunoff} always pick the same winners.} The numbers highlighted in red identify pairs of voting methods that  eliminate sure weak dominance manipulation according to Definition \ref{eliminate-sure-dominance}.   For instance, 25\% of the  $(3,4)$-profiles witness weak dominance manipulation for \texttt{Borda}, 15\% of the pointed $(3,4)$-profiles witness weak dominances manipulation for \texttt{StrictNanson}, but there are no instances of sure weak dominance manipulation for $\{\texttt{Borda}, \texttt{StrictNanson}\}$. 

To illustrate what happens with larger number of voters, we present in Figure \ref{SureManyVoters} percentages for sure weak dominance manipulation for \texttt{Borda} alone and \texttt{Borda} paired with several other voting methods. The data for $(4,4)$ was obtain by exhaustive search, while the data for $(4,5)$--$(4,70)$ was obtained by generating 10,000 profiles for each $(4,m)$ using the Impartial Culture model.

\begin{wrapfigure}{l}{0.55\textwidth}
\begin{center}
\includegraphics[scale=0.4]{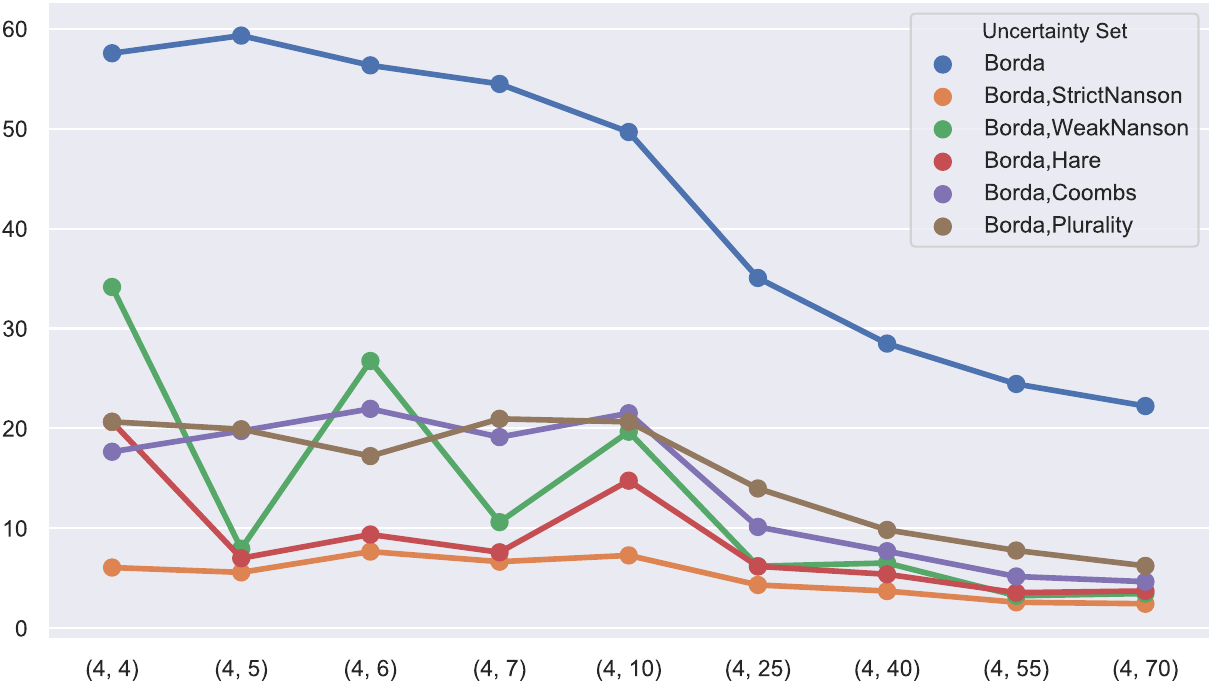}
\end{center}
\caption{{\small Percentage of $(4,m)$-profiles witnessing sure \newline weak dominance manipulation for \texttt{Borda} alone and \newline  \texttt{Borda} paired with several other methods.}}\label{SureManyVoters}
 \end{wrapfigure} 

By exhaustive search from $(3,4)$--$(3,8)$, we find that  $\{$\texttt{Borda}, \texttt{Baldwin}$\}$, $\{$\texttt{Borda}, \texttt{StrictNanson}$\}$, $\{$\texttt{WeakNanson}, \texttt{Baldwin}$\}$, and $\{$\texttt{WeakNanson}, \texttt{StrictNanson}$ \}$ are not susceptible to sure weak dominance manipulation, while each method individually is. This raises the question of whether one can prove that for all $(3,m)$, those sets of methods eliminate sure weak dominance manipulation.~Indeed, we will prove this result. 

\begin{lemma}\label{BordaNansonIndividually} For any $n\geq 3$ and $m\geq 4$,\footnote{The lemma can also be proved by similar reasoning for $n\geq 4$ and $m\geq 3$. However, for $n=3$ and $m=3$, there are no instances of weak dominance manipulation for these methods.} \texttt{Baldwin}, \texttt{Borda}, \texttt{StrictNanson}, and \texttt{WeakNanson} are each susceptible to sure weak dominance manipulation for $(n,m)$.
\end{lemma}

\begin{proof} Given a $(3,m)$-profile $\mathbf{P}$, let $\mathbf{P}\uplus \mathbf{P}_2$ be the $(3,m+2)$-profile that results from adding to $\mathbf{P}$ two fresh voters with rankings $a\, b \, c$ and $c \, b \, a$, respectively. The rankings of candidates by \texttt{Borda} score in $\mathbf{P}$ and $\mathbf{P}\uplus \mathbf{P}_2$ are the same; the sets of candidates with the lowest \texttt{Borda} scores in $\mathbf{P}$ and $\mathbf{P}\uplus \mathbf{P}_2$ are the same; and the sets of candidates with less than average (resp.~less than or equal to average) \texttt{Borda} scores in $\mathbf{P}$ and $\mathbf{P}\uplus \mathbf{P}_2$ are the same. Thus, the set of \texttt{Borda} (resp.~\texttt{Baldwin}, \texttt{StrictNanson}, \texttt{WeakNanson}) winners does not change from $\mathbf{P}$ and $\mathbf{P}\uplus \mathbf{P}_2$. It follows that if $\mathbf{P}$ witnesses $\Delta$-manipulation for \texttt{Borda} (resp.~\texttt{Baldwin}, \texttt{StrictNanson}, \texttt{WeakNanson}) by transitioning to $\mathbf{P}'$, then $\mathbf{P}\uplus \mathbf{P}_2$  witnesses $\Delta$-manipulation for \texttt{Borda} (resp.~\texttt{Baldwin}, \texttt{StrictNanson},  \texttt{WeakNanson}) by transitioning to $\mathbf{P}'\uplus \mathbf{P}_2$. Thus, to prove that one of these voting methods is susceptible to $\Delta$-dominance manipulation for $(3,m)$ for any $m\geq 4$, it suffices to show this for $(3,4)$ and $(3,5)$, as susceptibility for all other pairs $(3,m)$ then follows by the preceding observation about $\mathbf{P}\uplus \mathbf{P}_2$. As we verified using our Python script, there are indeed $(3,4)$ and $(3,5)$ profiles witnessing weak dominance manipulability for all four methods individually.

For any $(n,m)$-profile $\mathbf{P}$, let $\mathbf{P}^+$ be the $(n+1,m)$-profile that results from adding to $\mathbf{P}$ a fresh candidate at the bottom of every voter's ranking. The set of \texttt{Borda} (resp.~\texttt{Baldwin}, \texttt{StrictNanson}, \texttt{WeakNanson}) winners does not change from $\mathbf{P}$ to $\mathbf{P}^+$. Thus, to prove that one of these methods is susceptible to $\Delta$-dominance manipulation for $(n,m)$ for any $n\geq 3$, it suffices to show this for $(3,m)$, as above.\end{proof}

The following main theorems, proved in the Appendix, show that uncertainty about the voting method can  entirely eliminate sure weak dominance manipulation.

\begin{theorem}\label{MainThm} For any $m\geq 4$, the sets $\{\texttt{Borda}, \texttt{Baldwin}\}$ and $\{\texttt{Borda}, \texttt{StrictNanson}\}$ eliminate sure weak dominance manipulation for $(3,m)$.
\end{theorem}

\begin{theorem}\label{MainThm2} For any $m\geq 4$, the sets $\{\texttt{WeakNanson}, \texttt{Baldwin}\}$ and $\{\texttt{WeakNanson}, \texttt{StrictNanson}\}$ eliminate sure weak dominance manipulation for $(3,m)$.
\end{theorem}

By contrast, already for $(3,5)$, $\{\texttt{Baldwin},\texttt{Borda}, \texttt{StrictNanson}, \texttt{WeakNanson}\}$ is susceptible to sure \textit{optimistic} dominance and sure \textit{pessimistic} dominance. We give an example for optimistic dominance, leaving pessimistic dominance as an exercise.

\begin{example} Let $\mathbf{P}$ be the following $(3,5)$ profile:\\

\begin{minipage}{1in}
\begin{center}
\begin{tabular}{c|c|c|c|c}
$1$ &  $2$ & $3$  & $4$ & $5$\\\hline
$a$ & $a$ & $b$ & $c$ & $c$\\
$b$ & $b$ & $a$ & $b$ & $b$\\
$c$ & $c$ & $c$ & $a$ & $a$\\
\end{tabular}
\end{center}
\end{minipage}\hspace{.45in}\begin{minipage}{4.6in}
\texttt{Baldwin}, \texttt{Borda}, \texttt{StrictNanson}, and \texttt{WeakNanson} choose $\{b\}$.  If voter 1 changes her ranking to $a\ c\ b$, then $\{a,b,c\}$ is the winning set for all four methods.  Since $\{a,b,c\} >_{\mathbf{P}_1}^{Opt} \{b\}$,  $(\mathbf{P},1)$ witnesses sure optimist dominance manipulability for $\{\texttt{Baldwin},\texttt{Borda}, \texttt{StrictNanson},\texttt{WeakNanson}\}$. 
\end{minipage}
\end{example}

For another contrast to Theorems \ref{MainThm}-\ref{MainThm2}, when we increase to 4 candidates, $\{$\texttt{Baldwin}, \texttt{Borda}, \texttt{StrictNanson}, \texttt{WeakNanson}$\}$ is susceptible to sure weak dominance manipulation.

\begin{example}\label{43example} Consider the following $(4,3)$-profile $\mathbf{P}$:\\

\begin{minipage}{.5in}
\begin{center}
\begin{tabular}{c|c|c}
$1$ &  $2$ & $3$\\\hline
$a$  & $b$ & $c$\\
$b$  & $d$ & $a$\\
$c$  & $c$ & $b$\\
$d$  & $a$ & $d$ \\
\end{tabular}
\end{center}
\end{minipage}\hspace{.55in}\begin{minipage}{5in}
The \texttt{Borda} winning set is $\{b\}$, while the \texttt{Baldwin}, \texttt{StrictNanson}, and \texttt{WeakNanson} winning sets are all $\{a,b,c\}$. If voter 1 changes her ranking to $a\,d\,b\,c$, then the \texttt{Borda} winning set is $\{a,b\}$, while the \texttt{Baldwin}, \texttt{StrictNanson}, and \texttt{WeakNanson} winning sets are $\{a\}$. Since $\{a,b\} >_{\mathbf{P}_1}^{weak} \{b\}$ and $\{a\} >_{\mathbf{P}_1}^{weak} \{a,b,c\}$, $(\mathbf{P},1)$ witnesses sure weak dominance manipulability for $\{\texttt{Baldwin}, \texttt{Borda}, \texttt{StrictNanson},\texttt{WeakNanson}\}$.
\end{minipage}
\end{example}

\subsection{The failure of the Duggan-Schwartz theorem for sets of voting methods}

A natural analogue of the Duggan-Schwartz theorem \cite{DugganSchwartz:2000} for sure manipulation of sets of voting methods would state that for any $(n,m)$ with $n\geq 3$ and set $S$ of methods, each of which is non-imposed and has no nominator, $S$ is susceptible to sure optimistic or pessimistic dominance manipulation for $(n,m)$.  However, this statement is false, as shown by the following example verified by our Python script.

\begin{example} These sets of methods (which are non-imposed and have no nominator) eliminate sure optimistic and pessimistic dominance manipulation for $(3,6)$:
$\{$\texttt{Baldwin}, \texttt{Condorcet}$\}$, 
$\{$\texttt{Condorcet}, \texttt{Copeland}$\}$, 
$\{$\texttt{Condorcet}, \texttt{MaxMin}$\}$,
$\{$\texttt{Condorcet}, \texttt{StrictNanson}$\}$, and
$\{$\texttt{Condorcet}, \texttt{WeakNanson}$\}$. 
\end{example}

\subsection{Cases where three methods are needed for elimination} 

So far we have only considered sets of two voting methods. But in some cases three voting methods are needed to eliminate sure manipulation. We saw in Example \ref{43example} that $\{\texttt{Borda}, \texttt{Baldwin}\}$ is susceptible to sure weak dominance manipulation for $(4,3)$. However, adding $\texttt{Coombs}$ to $\{\texttt{Borda}, \texttt{Baldwin}\}$ eliminates sure weak dominance manipulation, as verified by our Python script (which also finds profiles witnessing sure weak dominance manipulability for $\{\texttt{Borda}, \texttt{Coombs}\}$ and $\{\texttt{Coombs}, \texttt{Baldwin}\}$).

\begin{fact} $\{\texttt{Borda}, \texttt{Coombs}, \texttt{Baldwin}\}$  eliminates sure weak dominance manipulation for $(4,3)$. 
\end{fact}

 \subsection{$(n,m)$ for which sure manipulation cannot be eliminated}
 
 So far we have focused on eliminating sure manipulation using uncertainty about the voting method. However, as the following result shows, eliminating sure manipulation is not always possible.
 
 \begin{proposition} For every $(n,m)$, there are $n'>n$ and $m'>m$ such that every nonempty subset of $\mathsf{Methods}\setminus\{\texttt{Condorcet}\}$ is susceptible to sure weak dominance manipulation for $(n',m')$.
 \end{proposition}
 
 \begin{proof} First, we claim that if there is one $(n,m)$-profile that witnesses sure weak dominance manipulability for every nonempty subset of $\mathsf{Methods}\setminus\{\texttt{Condorcet}\}$, then for all $k\in\mathbb{N}$, there is an $(n,m+24k)$-profile that does so. For any profile $\mathbf{P}$, let $\mathbf{P}\uplus\mathbf{P}_{24}$ be the result of adding 24 new voters to $\mathbf{P}$, one with each of the possible 24 rankings of $\{a,b,c,d\}$. It is easy to see that for any $f\in\mathsf{Methods}$, $f(\mathbf{P})=f(\mathbf{P}\uplus\mathbf{P}_{24})$. It follows that if $\mathbf{P}$ witnesses sure $\Delta$-dominance manipulability for $S$ by transitioning to $\mathbf{P}'$, then so does $\mathbf{P}\uplus\mathbf{P}_{24}$ by transitioning to $\mathbf{P}'\uplus\mathbf{P}_{24}$.   In addition, using the construction from $\mathbf{P}$ to $\mathbf{P}^+$ in the proof of Lemma \ref{BordaNansonIndividually}, we can increase the number $n$ of candidates, since for any $f\in\mathsf{Methods}$, $f(\mathbf{P})=f(\mathbf{P}^+)$. Now consider the following $(4,4)$-profile~$\mathbf{P}$:\\
 
 \begin{minipage}{1in}
\begin{center}
\begin{tabular}{c|c|c|c}
$1$ & $2$ & $3$ & $4$    \\\hline
$a$ & $b$ & $c$ & $c$ \\
$b$ & $d$ & $a$ & $a$   \\
$c$ & $c$ & $b$ & $b$  \\
$d$ & $a$ & $d$ & $d$  \\
\end{tabular}
\end{center}
\end{minipage}\hspace{.2in}\begin{minipage}{4.85in}
Every method in $\mathsf{Methods}\setminus\{\texttt{Condorcet}\}$ chooses $\{c\}$ as the set of winners. If voter 1 changes to $b\,d\, c\, a$, then the winning set for 
all  methods in ${\mathsf{Method}\setminus\{\texttt{Condorcet}\}}$ becomes $\{b,c\}$. Since $\{b,c\}>_{\mathbf{P}_1}^{weak}\{c\}$, $(\mathbf{P},1)$ witnesses sure weak dominance manipulability for any nonempty subset of $\mathsf{Method}\setminus\{\texttt{Condorcet}\}$.\qedhere
\end{minipage}\end{proof}

\subsection{Reduction without elimination}\label{ReductionSection}

Even when eliminating sure manipulation is not possible, one may hope to reduce it. For example, in Figures \ref{3437sureweak} and \ref{SureManyVoters}, there are a number of cases in which a pair of methods does not eliminate sure weak dominance manipulation but does reduce the number of profiles witnessing sure weak dominance manipulation relative to either method individually. This motivates the following notions.

\begin{definition}\label{LessSure} For any sets $S$ and $S'$ of voting methods, $S$ is \textit{less susceptible to sure} \textit{$\Delta$-manipulation than $S'$ for $(n,m)$-profiles} (resp.~\textit{pointed $(n,m)$-profiles}) iff there are fewer $(n,m)$-profiles (resp.~pointed $(n,m)$-profiles) witnessing sure $\Delta$-manipulation for $S$ than there are for $S'$.
\end{definition}

\begin{definition}\label{ImproveOnAll} A set $S$ of methods \textit{improves on all its subsets with respect to sure $\Delta$-manipulation for $(n,m)$} iff $S$ is less susceptible to sure $\Delta$-manipulation for $(n,m)$-profiles than any nonempty $S'\subsetneq S$.
\end{definition}

\noindent Such improvement is especially pronounced with optimistic and pessimistic dominance, as in Figure \ref{OptPessFig}.

\begin{figure}[h]
\includegraphics[scale=0.28]{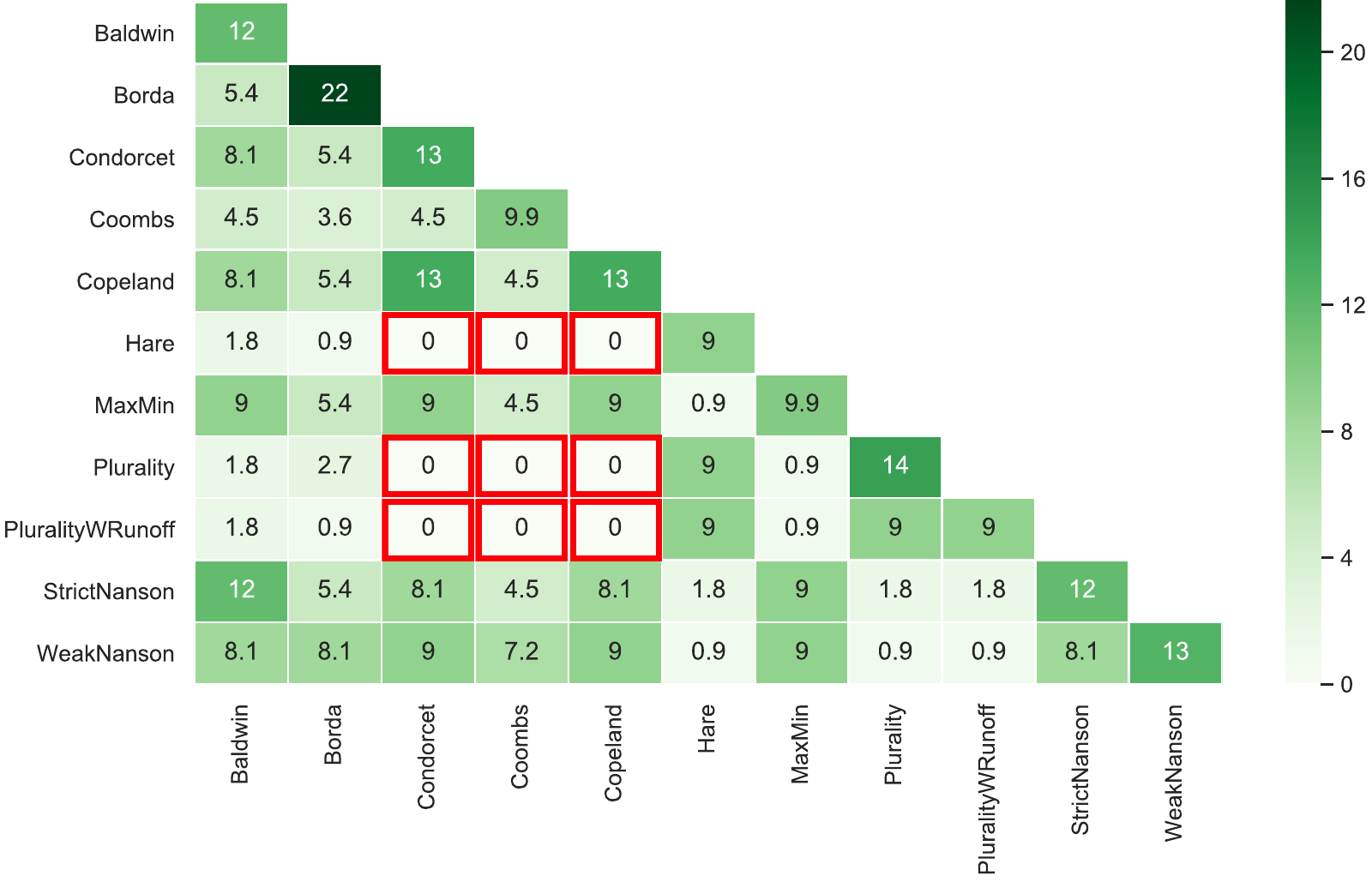}\quad\quad
\includegraphics[scale=0.28]{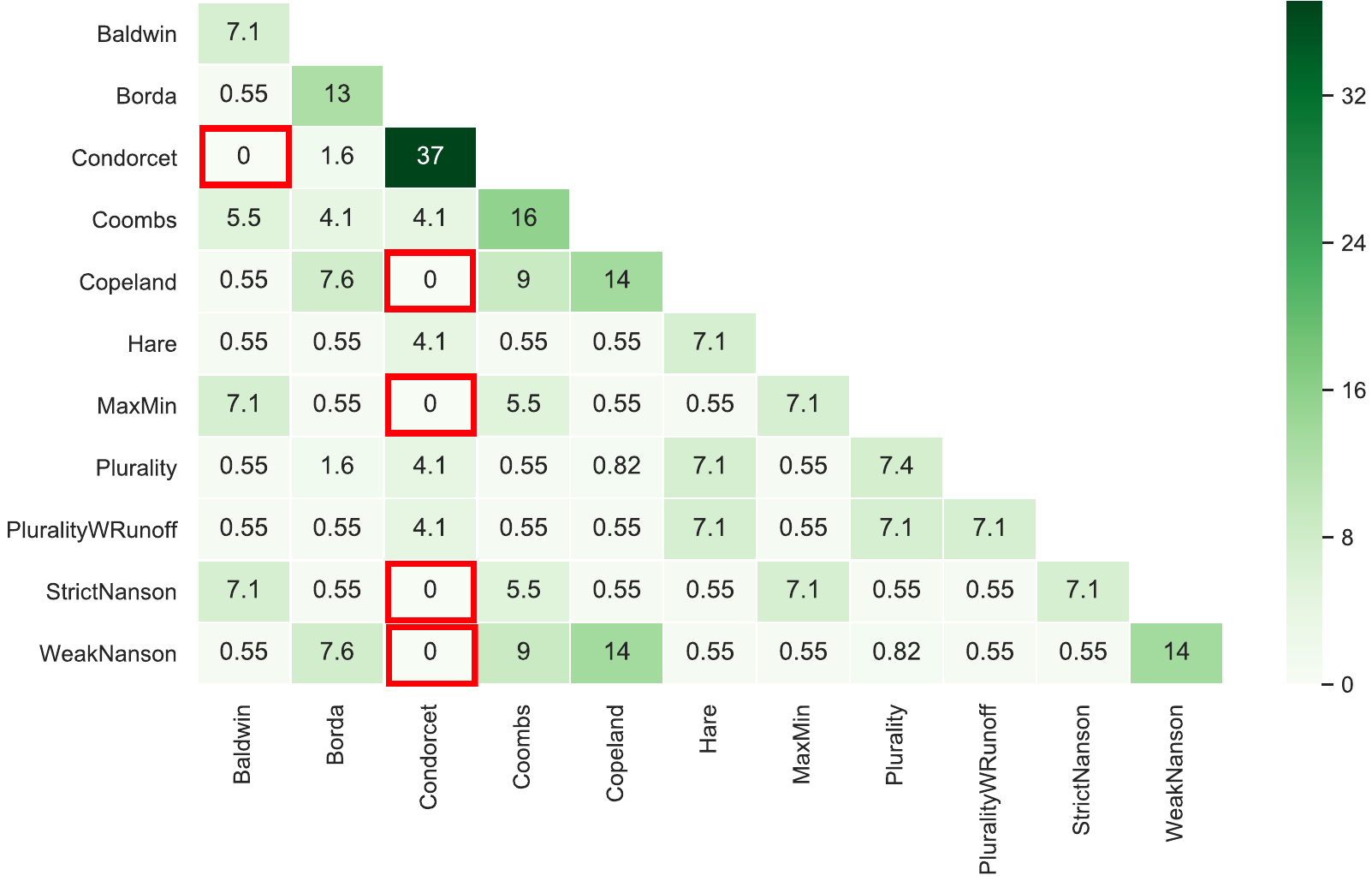}
\caption{{\small Percentage of profiles witnessing sure optimistic dominance manipulation for $(3,7)$ (left) and sure pessimistic dominance manipulation for $(3,6)$ (right).}}\label{OptPessFig}
\end{figure}

Even if a set $S$ does not improve on all of its subsets, that $S$ is less susceptible to sure $\Delta$-manipulation than one of its subset may still be significant. An election designer who intends to use method $f$ may wish to leave voters uncertain between $f$ and $f'$ in order to reduce the chance that voters will surely manipulate, relative to what would happen if voters knew the method was $f$, even if there is no reduction relative to what would happen if the planner intended to use $f'$ and voters knew this. For example, in Figure \ref{3437sureweak} for $(3,7)$, someone intending to use \texttt{Plurality} could reduce the percentage of profiles in which a voter will surely manipulate from 29\% to 9\% by leaving the voter uncertain between \texttt{Plurality} and \texttt{Hare}, even though an election designer intending to use \texttt{Hare} would have no incentive to do so, since the percentage of profiles witnessing sure manipulation for \texttt{Hare} by itself is already~9\%. 

A striking pattern in Figures \ref{3437sureweak}, \ref{SureManyVoters}, and \ref{OptPessFig} is that pairing \texttt{Borda} with another method leads to significant reductions in sure dominance manipulation. Using one or more methods to help reduce manipulation for a preferred method is even more important in the case of \textit{safe} manipulation discussed~next.

\section{Safe manipulation}\label{Section:SafeManipulation}

We now turn to a less conservative approach to strategic voting under uncertainty about the voting method: submit an insincere ranking whenever you know that doing so will lead to an outcome that is at least as good and might lead to a better outcome.     
 
\begin{definition} Let $(\mathbf{P},i)$ be a pointed profile, $\Delta$ a dominance notion, and $S$ a set of voting methods. Then $(\mathbf{P},i)$ \textit{witnesses safe $\Delta$-manipulability for $S$} iff there is a profile $\mathbf{P}'$ differing from $\mathbf{P}$ only in $i$'s ballot such that:
$ \forall f\in S: f(\mathbf{P}')\geq_{\mathbf{P}_i}^\Delta f(\mathbf{P}) \,\mbox{ and }\, \exists f\in S: f(\mathbf{P}')>_{\mathbf{P}_i}^\Delta f(\mathbf{P})$. We then say that $(\mathbf{P},i)$ witnesses safe $\Delta$-manipulability for $S$ \textit{by transitioning to $\mathbf{P}'$}. A profile $\mathbf{P}$ \textit{witnesses safe $\Delta$-manipulability for $S$} iff there is an $i\in V$ such that $(\mathbf{P},i)$ witnesses safe $\Delta$-manipulability for $S$.\end{definition}

\begin{remark}\label{HarmlessRemark} A variant of safe manipulability, which we will call \textit{harmless manipulability}, says that you should submit an insincere ranking whenever you know that doing so \textit{will not lead to a worse outcome} and might lead to a better outcome:
$\forall f\in S: f(\mathbf{P}')\not<_{\mathbf{P}_i}^\Delta f(\mathbf{P}) \,\mbox{ and }\, \exists f\in S: f(\mathbf{P}')>_{\mathbf{P}_i}^\Delta f(\mathbf{P})$. Since $f(\mathbf{P}') \not <_{\mathbf{P}_i}^{weak} f(\mathbf{P})$ does not imply $f(\mathbf{P}') \geq_{\mathbf{P}_i}^{weak} f(\mathbf{P})$, harmless weak dominance manipulability does not imply safe weak dominance manipulability for a given $(\mathbf{P},i)$ transitioning to $\mathbf{P}'$. But for pessimistic and optimistic dominance, harmless and safe manipulability are equivalent.
\end{remark} 

For certain profiles and uncertainty sets $S$, submitting an insincere ranking may lead to a better outcome with one method in $S$ but a worse outcome with another method in $S$, in which case it is not safe to manipulate using that insincere ranking.

\begin{example}\label{NotSafeEx}
Let $\mathbf{P}$ be  the following $(3,5)$-profile:\\

\begin{minipage}{1in}
\begin{center}
\begin{tabular}{c|c|c|c|c}
$1$ & $2$ & $3$ & $4$ & $5$\\\hline
$c$&  $a$ &  $b$ &  $c$ & $a$  \\
$b$&  $c$ &  $a$ &  $b$ &  $c$ \\
$a$ &  $b$ &  $c$ &  $a$ &  $b$\\
\end{tabular}
\end{center}
\end{minipage}\hspace{.5in}\begin{minipage}{4.55in}
Here $\texttt{Hare}(\mathbf{P})=\{a\}$, 
$\texttt{Borda}(\mathbf{P})=\{c\}$, and $\texttt{MaxMin}(\mathbf{P})=\{a, c\}$.  Suppose voter 1 changes her ranking to $b\ a\ c$, resulting in $\mathbf{P}'$. Then $\texttt{Hare}(\mathbf{P}')=\{b\}$,  $\texttt{Borda}(\mathbf{P}')=\{a\}$, and   $\texttt{MaxMin}(\mathbf{P}')=\{a, b\}$. Since $\{b\} >_\mathbf{P}^{weak} \{a\}$, voter 1 has an incentive to manipulate with \texttt{Hare}.    But voter 1 does not have an incentive to manipulate with \texttt{Borda}, since $\{c\}>^{weak}_{\mathbf{P}_1}\{a\}$, or \texttt{MaxMin}, since $\{a, b\}\not\geq^{weak}_{\mathbf{P}_1}\{a,c\}$. 
\end{minipage}\\

\noindent Thus, $(\mathbf{P}, 1)$ does not witness safe weak dominance manipulation for \texttt{Borda} together with any nonempty subset of $\{\texttt{Hare}, \texttt{MaxMin}\}$.
\end{example}

It seems too much to hope to \textit{eliminate} safe manipulation by adding reasonable methods to $S$; for this would require that for \textit{every} profile in which a manipulation results in a better outcome for one method in $S$, it results in a worse outcome for another method in $S$, which seems unlikely to hold for a set of reasonable methods. However, one can eliminate safe manipulation by adding methods that would be considered unreasonable by themselves.

\begin{example}\label{PivatoEx} For any distinct $x,y\in C$ and $i\in V$, let $f_{x,y,i}$ be the method such that $f_{x,y,i}(\mathbf{P})$ selects as the winner whichever of $x$ and $y$ is ranked higher according to $\mathbf{P}_i$ (cf.~\cite{nunez-pivato}). It is easy to see that if $S$ contains $f_{x,y,i}$ for each distinct $x,y\in C$, then $i$ cannot safely manipulate with $S$.
\end{example}

Although eliminating safe manipulation with reasonable methods may be too much to hope for, one can \textit{reduce} safe manipulation. Thus, we are  interested in the following analogue of Definition \ref{LessSure}.

  \begin{definition}\label{LessSafe} For any sets $S$ and $S'$ of voting methods, $S$ is \textit{less susceptible to safe} \textit{$\Delta$-manipulation than $S'$ for $(n,m)$-profiles} (resp.~\textit{pointed $(n,m)$-profiles}) iff there are fewer $(n,m)$-profiles (resp.~pointed $(n,m)$-profiles) witnessing sure safe $\Delta$-manipulation for $S$ than there are for $S'$.
\end{definition}

\begin{figure}[h]
\includegraphics[scale=0.28]{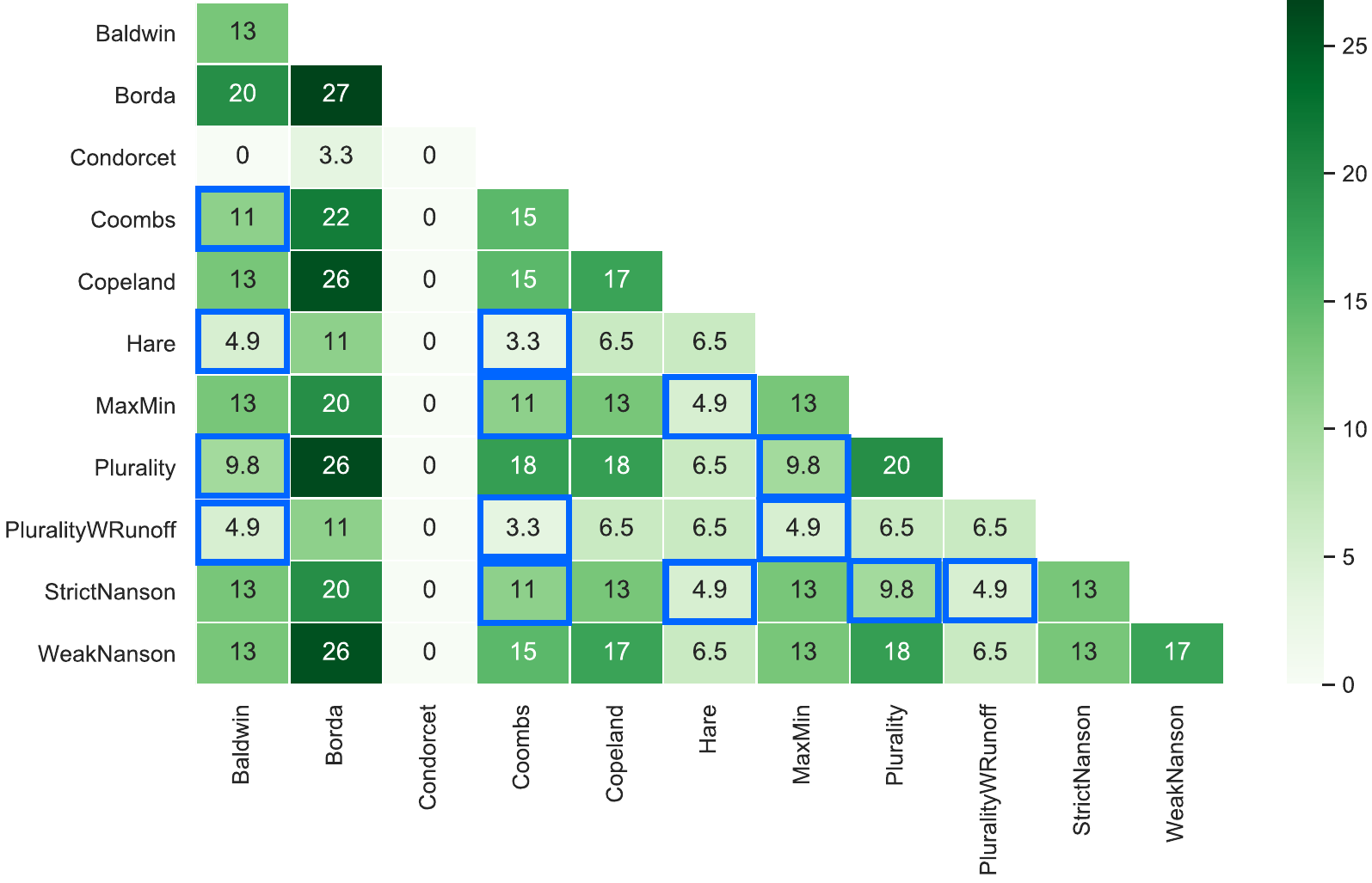}\quad\quad
\includegraphics[scale=0.28]{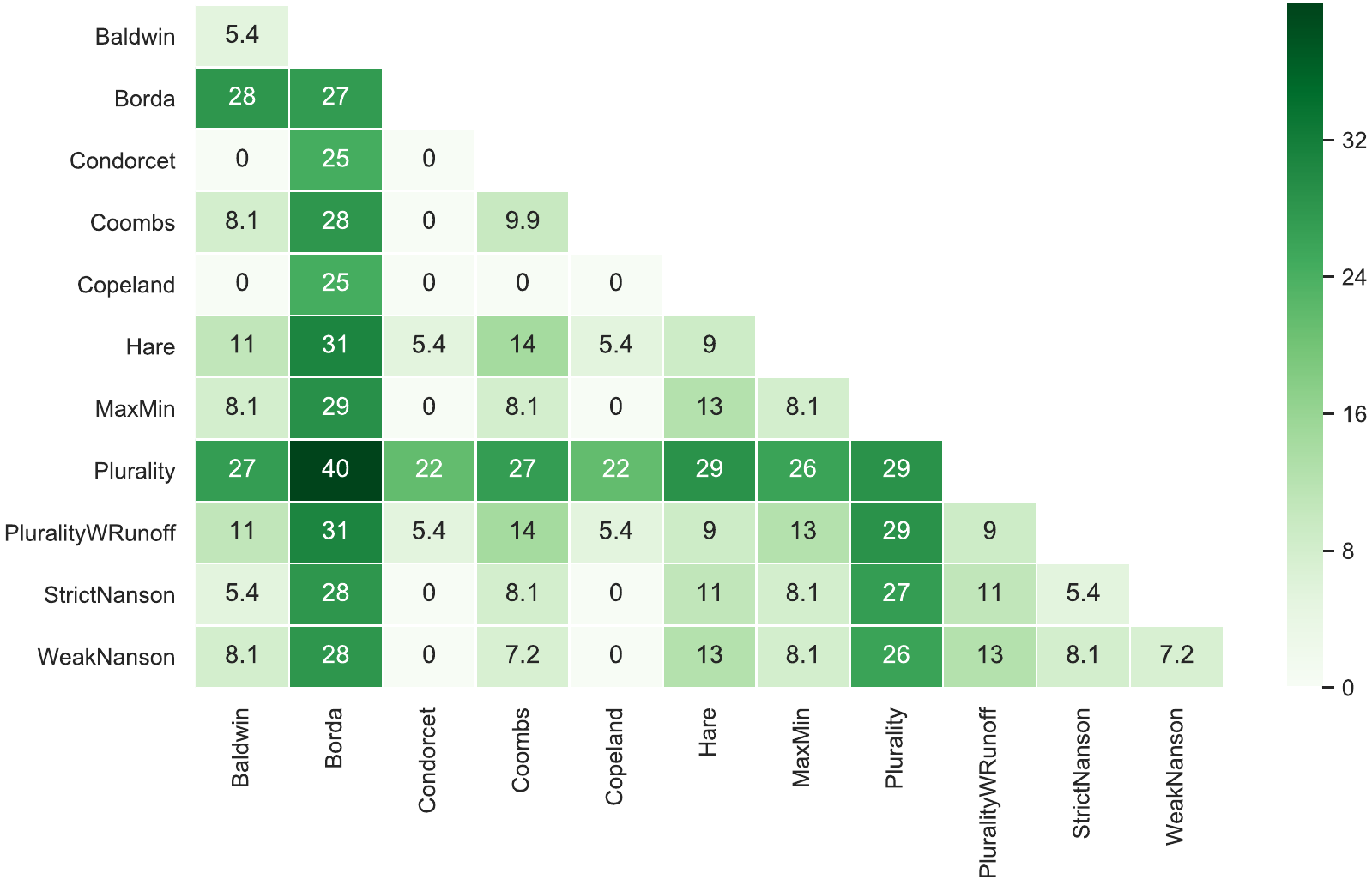}
\caption{{\small Percentage of profiles witnessing safe weak dominance manipulation for $(3,6)$ (left) and $(3,7)$ (right).}}\label{SafeHeat}
\end{figure}

Figure \ref{SafeHeat} shows the percentage of profiles witnessing safe weak dominance manipulation for $(3,6)$ and $(3,7)$ for sets of two voting methods. All of the following can happen: (1) Unlike with sure manipulation, with safe manipulation a voter who is uncertain between methods $f$ and $f'$ may have an incentive to manipulate on \textit{more} profiles than a voter who knows the method is $f$ and \textit{more} profiles than a voter who knows the method is $f'$. E.g., this happens for $(3,7)$ when $f=\mathtt{Borda}$ and $f'=\mathtt{Hare}$. (2) A voter who is uncertain between methods $f$ and $f'$ may have an incentive to manipulate on \textit{fewer} profiles than a voter who knows the method is $f$ and \textit{fewer} profiles than a voter who knows the method if $f'$. E.g., this happens for $(3,6)$ when $f=\mathtt{Coombs}$ and $f'=\mathtt{Hare}$. In Figure \ref{SafeHeat}, all examples of this phenomenon are indicated with the blue boxes. For $(3,7)$, there are no examples.  \begin{wrapfigure}{l}{0.55\textwidth} \begin{center}
\includegraphics[scale=0.4]{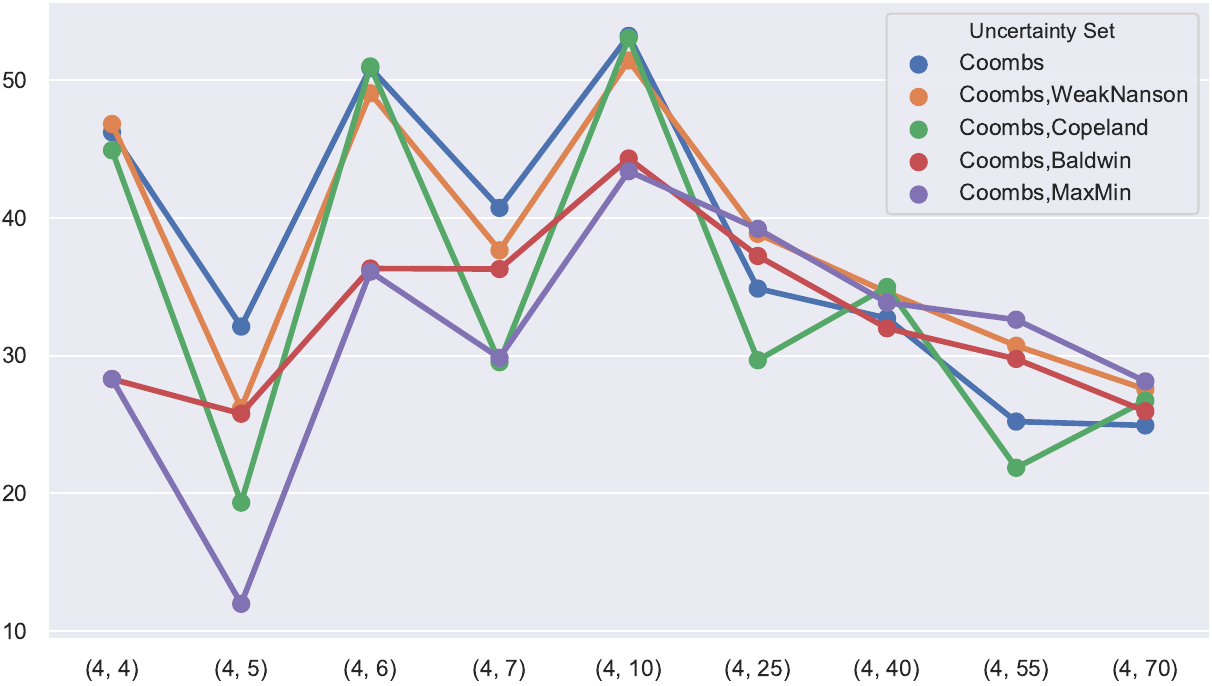}
\caption{{\small Percentage of $(4,m)$-profiles witnessing safe \newline weak dominance manipulation for \texttt{Coombs} alone and \newline  \texttt{Coombs} paired with several other methods. \newline $\,$}}\label{SafeManyVoters}
\end{center}
\end{wrapfigure} 
(3) A voter who is uncertain between methods $f$ and $f'$ may have an incentive to manipulate on \textit{fewer} profiles than a voter who knows the method is $f$ but \textit{more} profiles than a voter who knows the method is $f'$. E.g., this happens with $(3,6)$ when $f=\mathtt{Borda}$ and $f'=\mathtt{Hare}$.  In case (3), an election designer who intends to use method $f$ may decide to leave voters uncertain between $f$ and $f'$ in order to decrease the chance that voters will safely manipulate (cf.~the end of Section \ref{ReductionSection}).\footnote{In this case, a sophisticated voter with access to the manipulation data for $f$, $f'$, and $\{f,f'\}$ could infer that the election designer intenders to use $f$. For example, if $S=\{\mathtt{Borda},\mathtt{Hare}\}$ for $(3,6)$, then such a voter could infer that the designer intends to use \texttt{Borda}, since a designer who intends to use \texttt{Hare} will not decrease manipulation with $S=\{\mathtt{Borda},\mathtt{Hare}\}$. If the designer anticipates that voters are sophisticated in this way, then she should move from a single method to a pair only in case (2).} E.g., Figure \ref{SafeManyVoters} shows how an election designer intending to use \texttt{Coombs} could pair \texttt{Coombs} with several other methods to form an uncertainty set of two methods in order to decrease the percentage of profiles in which a voter will safely manipulate.

\section{Expected manipulation}\label{Section:ExpectedManipulation}

Our last approach to strategic voting under uncertainty about the voting method is the most liberal: assuming one's uncertainty about the voting method is given by a lottery on the set of voting methods, submit an insincere ranking if and only if doing so is more likely to lead to a better outcome than to lead to a worse outcome. Recall that a \textit{lottery} on a set $Y$ is a function $\nu : Y\to [0,1]$ such that $\sum_{y\in Y}\nu(y)=1$. For $X\subseteq Y$, let $\nu(X)=\sum_{x\in X}\nu(x)$.

\begin{definition} Let $(\mathbf{P},i)$ be a pointed profile, $\Delta$ a dominance notion, $S$ a set of voting methods, and $\nu$ a lottery on $S$. Then $(\mathbf{P},i)$ \textit{witnesses $\nu$-expected $\Delta$-manipulability for $S$} if and only if there is a profile $\mathbf{P}'$ differing from $\mathbf{P}$ only in $i$'s ballot such that
$
\nu (\{f\in S\mid  f(\mathbf{P}')>_{\mathbf{P}_i}^\Delta f(\mathbf{P})\} ) > \nu(\{f\in S\mid  f(\mathbf{P}')<_{\mathbf{P}_i}^\Delta f(\mathbf{P})\} )$.
Then we say $(\mathbf{P},i)$ witnesses $\nu$-expected $\Delta$-manipulability for $S$ \textit{by transitioning to $\mathbf{P}'$}. \end{definition}

For simplicity, here we focus on $\nu$ being the uniform lottery on $S$, in which case we simply speak of `expected $\Delta$-manipulability' instead of `$\nu$-expected $\Delta$-manipulability'. For $\nu$ uniform, this amounts to simply counting the number of voting methods in $S$ that lead to a better outcome vs. a worse outcome.

\begin{example}\label{ExpectedEx} The pointed profile $(\mathbf{P},1)$ in Example \ref{NotSafeEx} does not witness expected weak dominance manipulability for $\{\mathtt{Borda},\mathtt{Hare}\}$ by transitioning to $\mathbf{P}'$ since voter 1's manipulation results in a better outcome according to one method (\texttt{Hare}) but a worse outcome according to another method (\texttt{Borda}). However, like \texttt{Hare},  \texttt{Baldwin}  chooses $\{a\}$ as the set of winners in $\mathbf{P}$ and $\{b\}$ as the set in $\mathbf{P}'$. Hence $(\mathbf{P},1)$ witnesses expected weak dominance manipulability for $\{\mathtt{Baldwin},\mathtt{Borda},\mathtt{Hare}\}$  by transitioning to $\mathbf{P}'$, as two methods lead to a better outcome and only one leads to a worse outcome.\end{example}

The following fact is immediate from the definitions.

\begin{fact}\label{SafeToExpected} If $(\mathbf{P},i)$ witnesses safe $\Delta$-dominance manipulability for $S$, then $(\mathbf{P},i)$ witnesses expected $\Delta$-dominance manipulability for $S$. 
\end{fact}

The converse of Fact \ref{SafeToExpected} does not hold, as shown by Example \ref{ExpectedEx}. But for $|S|=2$ expected weak dominance manipulability is equivalent to harmless weak dominance manipulability (recall Remark \ref{HarmlessRemark}).

Using the obvious notion of \textit{less susceptible to expected $\Delta$-dominance manipulation} analogous to Definitions \ref{LessSure} and \ref{LessSafe}, we have the following result, inspired by \cite{nunez-pivato}, showing how adding a method to the set $S$ may reduce expected manipulation. 

\begin{proposition}\label{ReduceExpected} Suppose $S$ is a set of voting methods such that for some $f\in S$, $(n,m)$, pointed $(n,m)$-profile $(\mathbf{P},i)$, and $a,b\in C$, (1) $(\mathbf{P},i)$ witnesses weak dominance manipulability for $f$, and (2) for any $\mathbf{P}'$ differing from $\mathbf{P}$ only in $i$'s ranking, if $(\mathbf{P},i)$ witnesses weak dominance manipulability for $f$ by transitioning to $\mathbf{P}'$, then $\mathbf{P}'$ differs from $\mathbf{P}$ in $i$'s ranking of $a$ vs. $b$, and for all $g\in S\setminus \{f\}$, $g(\mathbf{P})=g(\mathbf{P}')$. 
Then where $f_{a,b,i}$ is the method defined in Example \ref{PivatoEx}, $S\cup \{f_{a,b,i}\}$ is less susceptible to expected weak dominance manipulation than $S$ for pointed $(n,m)$-profiles.
\end{proposition}

\begin{proof} The key properties of $f_{a,b,i}$ are that (i) there are no profiles $\mathbf{P}$ and $\mathbf{P}'$ differing only in $i$'s ranking such that we have $f_{a,b,i}(\mathbf{P}')>^{weak}_{\mathbf{P}_i} f_{a,b,i}(\mathbf{P})$, and (ii) if $\mathbf{P}$ and $\mathbf{P}'$ differ in $i$'s ranking of $a$ vs.~$b$, then $f_{a,b,i}(\mathbf{P}')<^{weak}_{\mathbf{P}_i} f_{a,b,i}(\mathbf{P})$. It follows from (i) that any pointed profile witnessing expected weak dominance manipulability for ${S\cup \{f_{a,b,i}\}}$ also witnesses expected weak dominance manipulability for $S$. Thus, to prove the proposition, we need only find a pointed profile witnessing expected weak dominance manipulability for  $S$ but not ${S\cup \{f_{a,b,i}\}}$. By assumption, there is a pointed $(n,m)$-profile $(\mathbf{P},i)$ and $a,b\in C$ satisfying items (1) and (2) of the proposition. Given (1), consider any $\mathbf{P}'$ such that $(\mathbf{P},i)$ witnesses weak dominance manipulability for $f$ by transitioning to $\mathbf{P}'$. It then follows by (2) that $(\mathbf{P},i)$ witnesses \textit{expected} weak dominance manipulability for $S$ by transitioning to $\mathbf{P}'$. However, we claim that $(\mathbf{P},i)$ does not witness expected weak dominance manipulability for $S\cup \{f_{a,b,i}\}$ by transitioning to $\mathbf{P}'$. This follows from the facts that $f$ and $f_{a,b,i}$ are equally likely according to the uniform measure, $\{h\in S\cup \{f_{a,b,i}\}\mid h(\mathbf{P}')<^{weak}_{\mathbf{P}_i}h(\mathbf{P})\}=\{f_{a,b,i}\}$ by (2) and (ii), and $\{h\in S\cup \{f_{a,b,i}\}\mid h(\mathbf{P}')>^{weak}_{\mathbf{P}_i}h(\mathbf{P})\}=\{f\}$ by (2).\end{proof}

We conclude this section with an example in which the conditions of Proposition \ref{ReduceExpected} apply.

\begin{example}Consider the following $(3,5)$-profile $\mathbf{P}$:\\

\begin{minipage}{1in}
\begin{center}
\begin{tabular}{c|c|c|c|c}
$1$ & $2$ & $3$ & $4$ & $5$\\\hline
$c$&  $c$ &  $a$ &  $a$ & $a$  \\
$a$&  $a$ &  $c$ &  $c$ &  $c$ \\
$b$ &  $b$ &  $b$ &  $b$ &  $b$\\
\end{tabular}
\end{center}
\end{minipage}\hspace{.5in}\begin{minipage}{4.55in}
Let $S=\{\texttt{Borda},\texttt{Coombs}\}$. $(\mathbf{P},i)$ witnesses the weak manipulability of $\texttt{Borda}$ by transitioning to the profile $\mathbf{P'}$ in which $i$'s new ranking is $c \, b \, a$, as the $\texttt{Borda}$  winning set in $\mathbf{P}$ is $\{a\}$, the \texttt{Borda} winning set in $\mathbf{P}'$ is $\{a,c\}$, and $\{a,c\}>^{weak}_{\mathbf{P}_1}\{a\}$. 
\end{minipage}\\

\noindent Moreover, this is the only $\mathbf{P}'$ differing from $\mathbf{P}$ only in $i$'s ranking such that $(\mathbf{P},i)$ witnesses the weak manipulability of $\texttt{Borda}$ by transitioning to $\mathbf{P'}$. Finally, the winning set for \texttt{Coombs} in both $\mathbf{P}$ and $\mathbf{P}'$ is $\{a\}$. Thus, the conditions of Proposition \ref{ReduceExpected} are satisfied for $f=\mathtt{Borda}$, so $\{\texttt{Borda},\texttt{Coombs},f_{a,b,i}\}$ is less susceptible to expected weak dominance manipulation than $\{\texttt{Borda},\texttt{Coombs}\}$.\end{example}

\section{Relation to probabilistic social choice}\label{Section:ProbabilisticSocialChoice}

In this section, we briefly relate our work to strategic voting in the setting of probabilistic social choice. 

\begin{definition} A \textit{probabilistic social choice function} (PSCF) is a function $F$ assigning to each profile $\mathbf{P}$ a lottery $F(\mathbf{P})$ on $C$.
\end{definition}

To define manipulation of PSCFs, we need a notion of when a voter prefers one lottery to another. Among many possible options (see, e.g., \cite[Sec.~1.3.2]{brandt-survey}), the following is popular.

\begin{definition} Let $\mu$ and $\mu'$ be lotteries on $C$ and $(\mathbf{P},i)$ a pointed profile. We say that $\mu$ \textit{stochastically dominates} $\mu'$ in $(\mathbf{P},i)$ if and only if for every $x\in C$, the probability that $\mu$ selects a candidate ranked at least as highly as $x$ by $i$ is greater than or equal to the probability that $\mu'$ selects a candidate ranked at least as highly as $x$ by $i$:
\[\forall x\in C: \underset{y\,:\, y\mathbf{R}_i x}{\sum} \mu(y)\geq \underset{y \,:\, y\mathbf{R}_i x}{\sum} \mu'(y).\]
We write $\mu \succsim_{\mathbf{P}_i} \mu'$ if $\mu$ stochastically dominates $\mu'$ in $(\mathbf{P},i)$ and $\mu\succ_{\mathbf{P}_i}\mu'$ if $\mu \succsim_{\mathbf{P}_i} \mu'$ but $\mu' \not\succsim_{\mathbf{P}_i} \mu$.
\end{definition}

Note that $\mu \succsim_{\mathbf{P}_i} \mu'$ if and only if for \textit{every}  utility function on $C$ that is compatible with $\mathbf{P}_i$, the expected utility of $\mu$ is at least as great as the expected utility of $\mu'$ (see, e.g., \cite[p.~302-3]{Bogomolnaia2001}).

\begin{definition}\label{SD} Let $F$ be a PSCF. A pointed profile $(\mathbf{P},i)$ \textit{witnesses stochastic dominance manipulability for $F$} if and only if there is a profile $\mathbf{P}'$ differing only in $i$'s ranking such that $F(\mathbf{P}')\succ_{\mathbf{P}_i} F(\mathbf{P})$.\footnote{Gibbard's \cite{gibbard1977} notion of srategyproofness for a PSCF $F$ is equivalent to the condition (called \textit{strong SD-strategyproofness} in \cite{brandt-survey}) that for every pointed profile $(\mathbf{P},i)$ and profile $\mathbf{P}'$ differing only in $i$'s ranking, $F(\mathbf{P})\succsim_{\mathbf{P}_i} F(\mathbf{P}')$, i.e., there is no profile $(\mathbf{P},i)$ witnessing manipulation in the sense that there exists a profile $\mathbf{P}'$ differing only in $i$'s ranking such that $F(\mathbf{P})\not\succsim_{\mathbf{P}_i} F(\mathbf{P}')$ (which Gibbard states in the equivalent form: there is some utility function on $C$ such that the expected utility of $F(\mathbf{P}')$ is greater than that of $F(\mathbf{P})$). Following Brandt \cite{brandt-survey}, we prefer the notion of $(\mathbf{P},i)$ witnessing stochastic dominance manipulation in Definition \ref{SD}, which is the notion used in Brandt's definition of (\textit{weak}) \textit{SD-strategyproofness}.} We then say that $(\mathbf{P},i)$ witnesses stochastic dominance manipulability for $F$ \textit{by transitioning to $\mathbf{P}'$}.
\end{definition}
\noindent For results on stochastic dominance manipulability, see \cite{brandt-survey}.

To relate the above notions to this paper, we observe how any set $S$ of voting methods gives rise to a PSCF, assuming that (i) each method in $S$ is equally likely to be used and (ii) each candidate in the set of winners selected by a method is equally likely to be chosen as the unique winner by a tiebreaking mechanism. Then the probability that a given candidate $a\in C$ will be chosen as the winner is:
\begin{eqnarray*}Pr(a \mbox{ wins})&=& \underset{f\in S}{\sum} Pr(a \mbox{ wins}\,|\, f\mbox{ is used}) \times Pr(f\mbox{ is used}) \\
&=& \underset{\underset{a\in f(\mathbf{P})}{f\in S}}{\sum} Pr(a \mbox{ wins}\,|\, f\mbox{ is used}) \times  Pr(f\mbox{ is used}) + \underset{\underset{a\not\in f(\mathbf{P})}{f\in S}}{\sum} Pr(a \mbox{ wins}\,|\, f\mbox{ is used}) \times Pr(f\mbox{ is used})\\
&=& \Big(\underset{\underset{a\in f(\mathbf{P})}{f\in S}}{\sum} \frac{1}{|f(\mathbf{P})|} \times  \frac{1}{|S|}\Big) +
 \Big(\underset{\underset{a\not\in f(\mathbf{P})}{f\in S}}{\sum} 0 \times \frac{1}{|S|}\Big) \quad = \quad\underset{\underset{a\in f(\mathbf{P})}{f\in S}}{\sum} \frac{1}{|f(\mathbf{P})|} \times  \frac{1}{|S|}.
\end{eqnarray*}
Thus, for a profile $\mathbf{P}$ and set $S$ of voting methods, we define the lottery $\mu^\mathbf{P}_S$ on $C$ by
\[\mu^\mathbf{P}_S(a) = \underset{\underset{a\in f(\mathbf{P})}{f\in S}}{\sum} \frac{1}{|f(\mathbf{P})|} \times  \frac{1}{|S|}.\]
Finally, for any set $S$ of voting methods, we define the PSCF $F_S$ by $F_S(\mathbf{P})= \mu^\mathbf{P}_S$. 
\begin{remark} For any set $S$ of voting methods and lottery $\nu$ on $S$, one can define a PSCF $F_{S,\nu}$ in the obvious way by weighting the methods in $S$ according to $\nu$ (and one can modify assumption (ii) with a non-uniform measure for tiebreaking). For simplicity, here we focus on the uniform measure on $S$.\end{remark}

The following fact can be verified from the definitions.

\begin{fact} If $(\mathbf{P},i)$ witnesses safe weak dominance manipulation for $S$ by transitioning to $\mathbf{P}'$, then $(\mathbf{P},i)$ witnesses stochastic dominance manipulation for $F_S$ by transitioning to $\mathbf{P}'$.
\end{fact}

However, stochastic manipulation does not imply safe weak manipulation.

\begin{example} 
Consider the following $(3,4)$-profile $\mathbf{P}$ for $C=\{a, b, c\}$ and $V=\{1,2,3,4\}$:\\

\begin{minipage}{1in}
\begin{center}
\begin{tabular}{c|c|c|c}
$1$ & $2$ & $3$ & $4$ \\\hline
$a$ & $a$ & $b$ & $b$\\
$b$ & $c$ & $a$ & $a$\\
$c$ & $b$ & $c$ & $c$\\
\end{tabular}
\end{center}
\end{minipage}\hspace{.2in}\begin{minipage}{4.85in}
The \texttt{Coombs}, \texttt{Copeland}, and \texttt{Hare} winning set is $\{a, b\}$.   If voter 1 changes to the ranking $c\, a\, b$, transitioning to the profile $\mathbf{P}'$, then the new winning set for \texttt{Coombs} and \texttt{Copeland} is $\{a\}$ and the new winning set for \texttt{Hare} is $\{b\}$.  Note that $\{a\} >_{\mathbf{P}_1}^{weak} \{a,b\}$ and $\{a, b\} >_{\mathbf{P}_1}^{weak} \{b\}$, so voter 1's new ranking leads to a better weak dominance outcome according to \texttt{Coombs} and \texttt{Copeland} but a worse weak dominance outcome according to \texttt{Hare}.  
\end{minipage}\\

\noindent Thus, $(\mathbf{P}, 1)$ does not witness safe weak dominance manipulability by transitioning to $\mathbf{P}'$.  

Let $S=\{\texttt{Coombs}, \texttt{Copeland}, \texttt{Hare}\}$.  The  lottery $F_S(\mathbf{P})$ is [$a:1/2$, $b: 1/2$, $c:0$]. The  lottery $F_S(\mathbf{P}')$ is
[$a:2/3$, $b: 1/3$, $c:0$]. As voter $1$'s ranking in $\mathbf{P}$ is $a\,b\,c$, $F_S(\mathbf{P}')$  stochastically dominates $F_S(\mathbf{P})$, since the probability according to  $F_S(\mathbf{P}')$ of selecting a candidate at least as good as $b$ is equal to the probability according to $F_S(\mathbf{P})$ (both are 1), and the probability according to $F_S(\mathbf{P}')$ of selecting a candidate at least as good as  $a$ (probability $2/3$) is  greater than the probability according to $F_S(\mathbf{P})$ (probability $1/2$). So $(\mathbf{P},1)$  witnesses stochastic dominance manipulability for $F_S(\mathbf{P})$ by transitioning~to~$\mathbf{P}'$.  \end{example}

It is easy to see abstractly that stochastic dominance manipulation also does not imply \textit{expected} weak dominance manipulation: e.g., if $S=\{f_1,f_2\}$ and the transition from $(\mathbf{P},i)$ to $\mathbf{P}$ is such that ${f_1(\mathbf{P})<_{\mathbf{P}_i}^{weak} f_1(\mathbf{P}')}$ and  $f_2(\mathbf{P})>_{\mathbf{P}_i}^{weak} f_2(\mathbf{P}')$, then the transition does not witness expected weak manipulation, but the amount by which $f_1$ increases the probability of getting a preferred candidate may be greater than the amount by which $f_2$ decreases the probability of getting a preferred candidate, so that the lottery $F_{f_1,f_2}(\mathbf{P}')$ stochastically dominates the lottery $F_{f_1,f_2}(\mathbf{P})$. The same idea applies for more than two methods. It remains to be seen whether an example of this kind exists with standard voting methods.

We can, however, use standard methods to show that expected weak dominance manipulation does not imply stochastic weak dominance manipulation.

\begin{example} Let $\mathbf{P}$ and $\mathbf{P}'$ be the profiles in Example \ref{NotSafeEx} and $S=\{\mathtt{Baldwin},\mathtt{Borda},\mathtt{Hare}\}$. In $\mathbf{P}$, \texttt{Baldwin} and \texttt{Hare} select $a$ as the winner, while \texttt{Borda} selects $c$ as the winner. Thus, the lottery $F_S(\mathbf{P})$ is
[$a:2/3$, $b:0$, $c:1/3$]. In $\mathbf{P}'$, \texttt{Baldwin} and \texttt{Hare} select $b$ as the winner, while \texttt{Borda} selects $a$ as the winner. Thus, the lottery $F_S(\mathbf{P}')$ is [$a:1/3$, $b: 2/3$, $c:0$].
As voter $1$'s ranking in $\mathbf{P}$ is $c\,b\,a$, $F_S(\mathbf{P}')$ does not stochastically dominate $F_S(\mathbf{P})$, since the probability according to $F_S(\mathbf{P}')$ of selecting a candidate at least as good as $c$ is not greater than or equal to the probability according to $F_S(\mathbf{P})$ of selecting a candidate at least as good as $c$. Thus, $(\mathbf{P},1)$ does not witness stochastic dominance manipulability for $F_S(\mathbf{P})$ by transitioning to $\mathbf{P}'$. However, as two of the three methods in $S$ lead to a better set of winners (in the sense of weak dominance) in $\mathbf{P}'$ than in $\mathbf{P}$ according to voter $1$'s ranking in $\mathbf{P}$, $(\mathbf{P},1)$ does witness expected weak dominance manipulability for $\{\mathtt{Baldwin},\mathtt{Borda},\mathtt{Hare}\}$ by transitioning to $\mathbf{P}'$.\end{example}

\section{Conclusion}\label{Section:Conclusion}

In this paper, we have shown that uncertainty about the voting method can be used as a barrier to manipulation. Considering such uncertainty led to three decision rules that voters may use to decide when to manipulate: sure, safe, and expected manipulation. Related issues arise when studying probabilistic voting methods, though Section \ref{Section:ProbabilisticSocialChoice} shows that our notions of sure, safe, and expected manipulation do not collapse to a standard notion of stochastic manipulation of probabilistic voting methods.

This initial study relied heavily on computer searches. Two natural next steps are (i)~to \textit{prove} additional possibility (or impossibility) theorems, like our Theorems \ref{MainThm}-\ref{MainThm2}, and (ii) to incorporate uncertainty about voting methods into the asymptotic analysis of strategic voting as the number of candidates or voters increases (see, e.g., \cite{Slinko2002a,Mossel2015}). Finally, a full analysis should take into account both uncertainty about the voting method and uncertainty about how others will vote.

\bibliographystyle{eptcs}
\bibliography{strategy}

\appendix

\section{Proof of Theorem \ref{MainThm}}

In this appendix, we prove Theorem \ref{MainThm}: for any $m\geq 4$, the sets $\{$\texttt{Borda}, \texttt{Baldwin}$\}$ and $\{$\texttt{Borda}, \texttt{StrictNanson}$\}$ eliminate sure weak dominance manipulation for $(3,m)$.

\begin{proof} Given Lemma \ref{BordaNansonIndividually}, we need only show that the two sets of methods are not susceptible to sure weak dominance manipulation for $(3,m)$. Toward a contradiction, suppose $(\mathbf{P},i)$ is a pointed $(3,m)$-profile witnessing sure weak dominance manipulation for either of the two sets by transitioning to a profile $\mathbf{P}'$. Suppose $i$'s ballot in $\mathbf{P}$ is $\alpha\beta\gamma$.

\textbf{Claim 1}: the \texttt{Borda} winning set in $\mathbf{P}'$ does not contain all three candidates. This follows by analyzing the possible sets of \texttt{Borda} winners in $\textbf{P}$. A three-way tie $\{\alpha,\beta,\gamma\}$ does not weakly dominate any of the following sets for $i$:  $\{\alpha\}$, $\{\beta\}$, $\{\alpha,\beta\}$, $\{\alpha,\gamma\}$, $\{\alpha,\beta,\gamma\}$. This leaves only the winning sets $\{\beta,\gamma\}$ and $\{\gamma\}$ to consider. But voter $i$ cannot manipulate so as to change the set of \texttt{Borda} winners from $\{\beta,\gamma\}$ to $\{\alpha,\beta,\gamma\}$; for the \texttt{Borda} scores of $\beta$ and $\gamma$ must still be tied in $\mathbf{P}'$, so $i$'s ranking in $\mathbf{P}'$ must be $\beta\gamma\alpha$, which decreases $\alpha$'s \texttt{Borda} score. In addition, voter $i$ cannot manipulate so as to change the set of \texttt{Borda} winners from $\{\gamma\}$ to $\{\alpha,\beta,\gamma\}$. For if $i$ ranks $\gamma$ higher in $\mathbf{P}'$, then $\gamma$'s \texttt{Borda} score increases by at least one, and no other candidate's \texttt{Borda} score increases by more than one, so  the set of winners is still $\{\gamma\}$; hence $i$'s ranking in $\mathbf{P}'$ must be $\beta\alpha\gamma$, but this fails to add $\alpha$ to the set of winners.

Label the candidates in order  of ascending  \texttt{Borda}  score in $\mathbf{P}'$ as $\varphi$, $\psi$, and $\chi$. Thus, by \textbf{Claim 1}, 
\begin{eqnarray}&&B'(\varphi)<B'(\psi)\leq B'(\chi)\mbox{ or }\label{B'ineq}\\
&&B'(\varphi)=B'(\psi)< B'(\chi),\label{B'ineq2}\end{eqnarray} 
where $B'$ indicates the  \texttt{Borda}  score in $\mathbf{P}'$. Let $B(\varphi),B(\psi),B(\chi)$ be the \texttt{Borda} scores of $\varphi$, $\psi$, and $\chi$ in $\mathbf{P}$.

\textbf{Claim 2}: $i$'s ranking changes from $\mathbf{P}$ to $\mathbf{P}'$ either by switching her 1st and 2nd placed candidates or by switching her 2nd and 3rd candidates; thus, one candidate's \texttt{Borda} score remains the same, and no candidate's \texttt{Borda} score changes by more than one point. Suppose for a contradiction that  $i$'s 3rd place candidate in $\mathbf{P}$ is her 1st place candidate in $\mathbf{P}'$ or her 1st place candidate in $\mathbf{P}$ is her 3rd place candidate in $\mathbf{P}'$. Then since $i$'s ranking in $\mathbf{P}$ is $\alpha\beta\gamma$, we have that $i$'s ranking in $\mathbf{P}'$ is either $\gamma\beta\alpha$, $\gamma\alpha\beta$, or $\beta\gamma\alpha$. But we claim $i$ does not have an incentive to transition to any of these rankings from the original ranking $\alpha\beta\gamma$. The first two transitions are such that the only candidate to increase in \texttt{Borda} score is $i$'s last place candidate, which never improves the winning set for $i$. The transition from $\alpha\beta\gamma$ to $\beta\gamma\alpha$ does not improve the winning set if the winning set in $\mathbf{P}$ contains $\alpha$; and if the  winning set in $\mathbf{P}$ is $\{\beta\}$, $\{\gamma\}$, or $\{\beta,\gamma\}$, again the transition from $\alpha\beta\gamma$ to $\beta\gamma\alpha$ does not change the winning set. Thus, we have a contradiction with the assumption that $(\mathbf{P},i)$ witnesses sure weak dominance manipulation for \texttt{Borda}. 

Using \textbf{Claim 2}, we can rule out case (\ref{B'ineq2}) above. For $(3,m)$, \texttt{Borda} is not \textit{single-winner manipulable} \cite[p.~57, Exercise 10]{Taylor2005}, which means that if $\{\chi\}$ is the set of  \texttt{Borda} winners in $\mathbf{P}'$, then the set of  \texttt{Borda} winners in $\mathbf{P}$ cannot be a singleton, so it must be one of $\{\varphi,\psi,\chi\}$, $\{\psi,\chi\}$, $\{\varphi,\chi\}$, or $\{\varphi,\psi\}$. By \textbf{Claim 2}, there is no way $i$ can change the \texttt{Borda}  scores from $B(\varphi)=B(\psi)=B(\chi)$ to $B'(\varphi)=B'(\psi)< B'(\chi)$, so we can rule out $\{\varphi,\psi,\chi\}$. Also by \textbf{Claim 2}, in order for $i$ to change the \texttt{Borda}  scores from $B(\varphi)<B(\psi)=B(\chi)$ to $B'(\varphi)=B'(\psi)< B'(\chi)$, her ranking must go from $\psi\chi\varphi$ to $\chi\psi\varphi$ or from $\varphi\psi\chi$ to $\varphi\chi\psi$, but then the new winning set $\{\chi\}$ is worse for $i$ than the original $\{\psi,\chi\}$. Thus, we can rule out $\{\psi,\chi\}$, and by the same reasoning, $\{\varphi,\chi\}$. Finally, by \textbf{Claim 2}, there is no way for $i$ to change the \texttt{Borda} scores from $B(\chi)<B(\varphi)=B(\psi)$ to $B'(\varphi)=B'(\psi)< B'(\chi)$, so we can rule out $\{\varphi,\psi\}$. Thus,  (\ref{B'ineq}) holds.

\textbf{Claim 3}: $\varphi$ has the unique below average \texttt{Borda} score in $\mathbf{P}$ and $\mathbf{P}'$. We argue by cases.

Case 1: $B'(\psi)=B'(\chi)$. Then it is immediate from (\ref{B'ineq}) that $\varphi$ has the unique below average \texttt{Borda} score in $\mathbf{P}'$. We now show that $\varphi$ has the unique below average \texttt{Borda} score in $\mathbf{P}$. Given $B'(\psi)=B'(\chi)$ and $B'(\varphi)+B'(\psi)+B'(\chi)=3m$,  we have
\begin{equation}B'(\varphi)+2B'(\psi)=3m.\label{2B'eq}\end{equation}
We claim that 
\begin{equation}B'(\psi)-B'(\varphi)  > 2.\label{ClaimEq}\end{equation}
 For $B'(\psi)-B'(\varphi)\neq 0$ since $\varphi$ was chosen as the candidate with the lowest \texttt{Borda} score in $\mathbf{P}'$, and if $B'(\psi)-B'(\varphi)=k$ for $k\in\{1,2\}$, so $B'(\varphi)=B'(\psi) -k$, then from (\ref{2B'eq}) we have $3B'(\psi)-k = 3m$, a contradiction. Now by \textbf{Claim 2}, we have $| B(\varphi)-B'(\varphi) | \leq 1$, $| B(\psi)-B'(\psi) | \leq 1$, and $| B(\chi)-B'(\chi) | \leq 1$. It follows by (\ref{ClaimEq}) that $B(\varphi) < B(\psi)$ and $B(\varphi)<B(\chi)$. Thus, $\varphi$ has a below average \texttt{Borda} score in $\mathbf{P}$. Finally, we claim that $\varphi$ is the only candidate with a below average \texttt{Borda} score in $\mathbf{P}$. Since the average \texttt{Borda} score is $m$, this means $B(\psi)\geq m$ and $B(\chi)\geq m$. Suppose for contradiction that $B(\psi)< m$ or $B(\chi)< m$. Without loss of generality, suppose $B(\psi)< m$. By \textbf{Claim 2} and the fact that $B'(\psi)=B'(\chi)$, we have $|B(\psi)-B(\chi)|\leq 2$. But together $B(\varphi)< B(\psi)$, $B(\psi)< m$, and $|B(\psi)-B(\chi)|\leq 2$ contradict the fact that $B(\varphi)+B(\psi)+B(\chi)=3m$. Thus, $\varphi$ has the unique below average \texttt{Borda} score in $\mathbf{P}$.
 
 Case 2: $B'(\psi)<B'(\chi)$. First, we claim it is not the case that $B(\varphi)=B(\psi)=B(\chi)$. It follows from \textbf{Claim 2} that there are only two ways to go from $B(\varphi)=B(\psi)=B(\chi)$ to ${B'(\varphi)<B'(\psi)<B'(\chi)}$: $i$'s ranking goes from $\varphi\chi\psi$ to $\chi\varphi\psi$ or from $\psi\varphi\chi$ to $\psi\chi\varphi$. But in both cases the new  \texttt{Borda} winning set $\{\chi\}$ does not weakly dominate the old winning \texttt{Borda} set $\{\varphi,\psi,\chi\}$, contradicting the assumption that $i$ had an incentive to manipulate. In addition, it is not the case that $B(\varphi)= B(\psi) < B(\chi)$ or $B(\varphi)< B(\psi) < B(\chi)$, for then $i$ would have no incentive to transition to $B'(\varphi)<B'(\psi)<B'(\chi)$, since the \texttt{Borda} winning set would not change. Thus, we have that $B(\varphi) < B(\psi)=B(\chi)$, so $\varphi$ has the unique below average \texttt{Borda} score in $\mathbf{P}$. Now by \textbf{Claim 2}, the \texttt{Borda} score of one of $\varphi$, $\psi$, and $\chi$ must remain the same from $\mathbf{P}$ to $\mathbf{P}'$. But we cannot have $B(\varphi)=B'(\varphi)$, for then in order to go from $B(\psi)=B(\chi)$ to $B'(\psi)<B'(\chi)$, by \textbf{Claim 2} $i$'s ranking must either change from $\psi\chi\varphi$ to $\chi\psi\varphi$ or from $\varphi\psi\chi$ to $\varphi\chi\psi$; but in both cases the new \texttt{Borda} winning set $\{\chi\}$ does not weakly dominate the old winning \texttt{Borda} set $\{\psi,\chi\}$, contradicting the assumption that $i$ had an incentive to manipulate. Thus, either $B(\psi)=B'(\psi)$ or $B(\chi)=B'(\chi)$. But if $B(\psi)=B'(\psi)$ or $B(\chi)=B'(\chi)$, then together $B(\psi)=B(\chi)$ and $B'(\psi)<B'(\chi)$ imply $|B'(\psi)- B'(\chi)|=1$ by \textbf{Claim 2}, in which case $B'(\varphi)<B'(\psi)<B'(\chi)$ implies that $\varphi$ has the unique below average \texttt{Borda} score in $\mathbf{P}'$. 
 
\textbf{Claim 4}: the majority ordering between $\psi$ and $\chi$ does not change from $\mathbf{P}$ to $\mathbf{P}'$. This follows from the claim that $i$'s ordering of $\psi$ and $\chi$ does not change from $\mathbf{P}$ to $\mathbf{P}'$. We know $\varphi$ is not in the set of \texttt{Borda} winners in $\mathbf{P}$ or in $\mathbf{P}'$ by \textbf{Claim 3}, and the set of \texttt{Borda} winners in $\mathbf{P}'$ is not $\{\varphi,\psi,\chi\}$ by \textbf{Claim 1}.

Case 1: the \texttt{Borda} winning set in $\mathbf{P}'$ is $\{\psi,\chi\}$, so $B'(\psi)=B'(\chi)$. Suppose for contradiction that $i$'s ranking for $\psi$ vs. $\chi$ in $\mathbf{P}$ is the reverse of $i$'s ranking in $\mathbf{P}'$. Suppose, without loss of generality, that $i$ ranks $\psi$ over $\chi$ in $\mathbf{P}$ but $i$ ranks $\chi$ over $\psi$ in $\mathbf{P}'$. Then since in $\mathbf{P}'$, we have $B'(\psi)=B'(\chi)$, it follows that $B(\psi)>B(\chi)$, which with \textbf{Claim 3} implies that $\psi$ is the unique \texttt{Borda} winner in $\mathbf{P}$. But then there is no incentive for $i$ to manipulate, as $i$ would do worse with the winning set $\{\psi,\chi\}$ in $\mathbf{P}'$. This contradicts our assumption that $(\mathbf{P},i)$ witnesses manipulation for \texttt{Borda} by transitioning to $\mathbf{P}'$.

Case 2: the \texttt{Borda} winning set in $\mathbf{P}'$ is $\{\chi\}$, so $B'(\psi)< B'(\chi)$. Since $\varphi$ is not in the \texttt{Borda}  winning set in $\mathbf{P}$ by \textbf{Claim 3}, the \texttt{Borda} winning set in $\mathbf{P}$ is  $\{\psi,\chi\}$ or $\{\psi\}$, but the latter is ruled out since \texttt{Borda} is not single-winner manipulable. Since by assumption the winning set in $\mathbf{P}'$ weakly dominates that in $\mathbf{P}$ for $i$, it follows that $i$ ranks $\chi$ above $\psi$ in $\mathbf{P}$. But then the transition from the \texttt{Borda} winning set $\{\psi,\chi\}$ to  $\{\chi\}$ cannot be the result of $i$ switching her ranking of $\chi$ and $\psi$ so that $\psi$ is ranked above $\chi$ in $\mathbf{P}'$.

For any $(3,m)$-profile, the set of \texttt{Baldwin} (resp.~\texttt{StrictNanson}) winners  may be obtained by first eliminating the candidates (if there are any) with the strictly lowest \texttt{Borda} scores (resp.~with below average \texttt{Borda} scores) and then selecting as winners from the remaining candidates those who are maximal in the majority ordering. Thus, by \textbf{Claim 3} and \textbf{Claim 4}, the sets of \texttt{Baldwin} winners and \texttt{StrictNanson} winners do not change from $\mathbf{P}$ to $\mathbf{P}'$. This contradicts the assumption that $(\mathbf{P},i)$ witnesses sure weak dominance manipulation for $\{\texttt{Borda}, \texttt{Baldwin}\}$ or $\{\texttt{Borda}, \texttt{StrictNanson}\}$ by transitioning to $\mathbf{P}'$.\end{proof}

\section{Proof of Theorem \ref{MainThm2}}

In this appendix, we prove Theorem \ref{MainThm2}: for any $m\geq 4$, $\{$\texttt{WeakNanson}, \texttt{Baldwin}$\}$ and $\{$\texttt{WeakNanson}, \texttt{StrictNanson}$\}$ eliminate sure weak dominance manipulation for $(3,m)$.

\begin{proof} Given Lemma \ref{BordaNansonIndividually}, we need only show that the two sets of methods are not susceptible to sure weak dominance manipulation for $(3,m)$. Consider an arbitrary pointed profile $(\mathbf{P},i)$. We show that for any profile $\mathbf{P}'$ differing from $\mathbf{P}$ only in $i$'s ranking, $(\mathbf{P},i)$ does not witness sure weak dominance manipulation for $\{\texttt{WeakNanson}, \texttt{Baldwin}\}$ or $\{\texttt{WeakNanson}, \texttt{StrictNanson}\}$ by transitioning to $\mathbf{P}'$. Let $i$'s ranking in $\mathbf{P}$ be $\alpha\beta\gamma$. Let $B(\alpha)$, $B(\beta)$, and $B(\gamma)$ be the \texttt{Borda} scores of $\alpha$, $\beta$, and $\gamma$, respectively, in $\mathbf{P}$, and likewise for  $B'(\alpha)$, $B'(\beta)$, and $B'(\gamma)$ in $\mathbf{P}'$. For any $(3,m)$-profile, the set of \texttt{WeakNanson} (resp.~\texttt{StrictNanson}, \texttt{Baldwin}) winners  may be obtained by first eliminating the candidates (if there are any) with less than or equal to average \texttt{Borda} scores (resp.~with below average \texttt{Borda} scores, with strictly lowest \texttt{Borda} scores) and then selecting as winners from the remaining candidates those who are maximal in the majority ordering. Note that the average $\texttt{Borda}$ score for any given $(3,m)$-profile is $m$. We consider the following exhaustive list of cases, where `WN' stands for \texttt{WeakNanson}:

1. $B(\beta),B(\gamma)\leq m < B(\alpha)$;  or $B(\beta)\leq m < B(\alpha), B(\gamma)$ and  $\alpha>^M_\mathbf{P}\gamma$; or $B(\gamma)\leq m < B(\alpha), B(\beta)$ and $\alpha>^M_\mathbf{P}\beta$. Then $\{\alpha\}$ is the WN-winning set in $\mathbf{P}$, so $i$ has no incentive to transition from $\mathbf{P}$ to~$\mathbf{P}'$

2. $B(\alpha)\leq m < B(\beta),B(\gamma)$ and $\beta>^M_\mathbf{P}\gamma$, so $\{\beta\}$ is the WN-winning set in $\mathbf{P}$. Then $B'(\alpha)\leq m<B'(\gamma)$, so $\alpha$ is not in the WN-winning set in $\mathbf{P}'$. But only sets that contain $\alpha$ weakly dominate $\{\beta\}$ for $i$.

3. $B(\alpha)\leq m < B(\beta),B(\gamma)$ and $\gamma>^M_\mathbf{P}\beta$, so $\{\gamma\}$ is the WN-winning set in $\mathbf{P}$. In this case $B'(\alpha)\leq m \leq  B'(\beta)$, $m< B'(\gamma)$, and $\gamma>^M_{\mathbf{P}'}\beta$ for any $\mathbf{P}'$ differing only in $i$'s ranking. So the winning set in $\mathbf{P}'$ is again $\{\gamma\}$, providing no  incentive under WN to transition to $\mathbf{P}'$.

4. $B(\alpha)\leq m < B(\beta),B(\gamma)$ and $\gamma=^M_\mathbf{P}\beta$, so $\{\beta,\gamma\}$ is the WN-winning set in $\mathbf{P}$. In this case $B'(\alpha)\leq m \leq  B'(\beta)$, $m< B'(\gamma)$, and $\gamma\geq ^M_{\mathbf{P}'}\beta$ for any $\mathbf{P}'$ differ only in $i$'s ranking. So the WN-winning set in $\mathbf{P}'$ is $\{\gamma\}$ or $\{\beta,\gamma\}$, neither of which weakly dominates $\{\beta,\gamma\}$ for $i$.

5. $B(\beta)\leq m < B(\alpha), B(\gamma)$ and  $\gamma \geq^M_\mathbf{P}\alpha$, so the WN-winning set  in $\mathbf{P}$ is either $\{\gamma\}$ or $\{\alpha,\gamma\}$. It follows that $B(\beta)\leq m-2$ and hence $B'(\beta)<m$, and $\gamma\geq ^M_{\mathbf{P}'}\alpha$ for any $\mathbf{P}'$ differ only in $i$'s ranking. Thus, if $\{\gamma\}$ is the WN-winning set in $\mathbf{P}$, it is the WN-winning set in $\mathbf{P}'$; and if $\{\alpha,\gamma\}$ is the WN-winning set in $\mathbf{P}$, then the WN-winning set in $\mathbf{P}'$ contains $\gamma$, but no set containing $\gamma$ weakly dominates $\{\alpha,\gamma\}$.

6. $B(\gamma)\leq m < B(\alpha), B(\beta)$, and $\beta\geq^M_\mathbf{P}\alpha$, so the WN-winning set is either $\{\beta\}$ or $\{\alpha,\beta\}$. It follows that $B(\gamma)< m$ and indeed $B(\gamma)\leq m-2$. Hence $B'(\gamma)\leq m$. Thus, $\gamma$ is eliminated for WN in the first round for $\mathbf{P}'$. Since $\beta\geq^M_\mathbf{P}\alpha$ with $i$ having the ranking $\alpha\beta\gamma$ in $\mathbf{P}$, it follows that $\beta\geq^M_{\mathbf{P}'}\alpha$ for any $\mathbf{P}'$ differing from $\mathbf{P}$ only in $i$'s ranking. Thus, if $\{\beta\}$ is the WN-wining set in $\mathbf{P}$, it is still the WN-winning set in $\mathbf{P}'$; and if $\{\alpha,\beta\}$ is the WN-winning set in $\mathbf{P}$, then the WN-winning set in $\mathbf{P}'$ contains $\beta$, but no set containing $\beta$ weakly dominates $\{\alpha,\beta\}$.

7.  $B(\alpha),B(\gamma)\leq m\leq B(\beta)$, so the WN-winning set in $\mathbf{P}$ is $\{\beta\}$ or $\{\alpha,\beta,\gamma\}$.  In this case $B'(\alpha)\leq m$, so the WN-winning set in $\mathbf{P}'$ is either $\{\alpha,\beta,\gamma\}$ or a set not containing $\alpha$, and in both cases the WN-winning set in $\mathbf{P}'$ does not weakly dominate the WN-winning set in $\mathbf{P}$.

8. $B(\alpha),B(\beta)\leq m<B(\gamma)$, so $\{\gamma\}$ is the WN-winning set in $\mathbf{P}$. Hence $B'(\alpha)\leq m <B'(\gamma)$. If $B'(\beta)\leq m$, then $\{\gamma\}$ is still the WN-winning set in $\mathbf{P}'$. So suppose $m<B'(\beta)$, which implies $B(\beta)=m$, $B(\alpha)<m$, and $B'(\alpha)<m$. There are only two post-manipulation rankings consistent with $m<B'(\beta)$: $\beta\alpha\gamma$ and $\beta\gamma\alpha$. Note that the majority ordering of $\beta$ and $\gamma$ has not changed from $\mathbf{P}$ to $\mathbf{P}'$, and in both $\mathbf{P}$ and $\mathbf{P}'$, $\alpha$ is the unique candidate with below average \texttt{Borda} score. It follows that the \texttt{StrictNanson} winning set in $\mathbf{P}$ and $\mathbf{P}'$ is the same, and the \texttt{Baldwin} winning set in $\mathbf{P}$ and $\mathbf{P}'$ is the same. Hence the transition from $\mathbf{P}$ to $\mathbf{P}'$ does not witness sure manipulation.  \qedhere
\end{proof}

\end{document}